\documentclass[reqno,xcolor=dvipsnames]{amsart}

\usepackage{geometry}
\usepackage[round,comma]{natbib}                

\usepackage{enumerate}
\usepackage{mathabx}

\usepackage{amsmath}
\usepackage{amsthm}
\usepackage{amsfonts}
\usepackage{amssymb}
\usepackage{dsfont}

\usepackage{nicefrac}
\usepackage{booktabs}
\usepackage{mathrsfs}
\usepackage{bm}
\usepackage{mathtools}

\newcommand{\ccA}{{\mathscr A}}\newcommand{\cA}{{\mathcal A}}
\newcommand{\ccB}{{\mathscr B}}

\newcommand{\cD}{{\mathcal D}}

\newcommand{\ccF}{{\mathscr F}}\newcommand{\cF}{{\mathcal F}}
\newcommand{\ccG}{{\mathscr G}}\newcommand{\cG}{{\mathcal G}}

\newcommand{\cK}{{\mathcal K}}

\newcommand{\ccN}{{\mathscr N}}

\newcommand{\cX}{{\mathcal X}}

\newcommand{\Ind}{{\mathds 1}}
\newcommand{\ind}[1]{\Ind_{\{#1\}}}

\newcommand{\restr}{\mathbf{\kern0.3ex
 \vert\kern-0.3ex}\backprime\kern0.3ex}

\newcommand{\FF}{\mathbb{F}}

\renewcommand{\cF}{\ccF}
\renewcommand{\cG}{\ccG}

\newtheorem{theorem}{Theorem}[section]
\newtheorem{lemma}[theorem]{Lemma}              
\newtheorem{proposition}[theorem]{Proposition}  

\theoremstyle{definition}
\newtheorem{example}[theorem]{Example} 
\newtheorem{definition}[theorem]{Definition} 
\newtheorem{remark}[theorem]{Remark}
\newtheorem{assumption}[theorem]{Assumption}

\usepackage[backgroundcolor=white,bordercolor=orange]{todonotes}

\usepackage[colorlinks,urlcolor=red,citecolor=blue,linkcolor=red]{hyperref}

\usepackage{subfiles} 

\newcommand{\psemac}{\mathfrak{P}_{\text{sem}}^{\text{ac}}}
\newcommand{\aset}{\mathcal{A}(t,x_t,\Theta)}
\DeclareMathOperator{\proj}{proj}

\usepackage{savesym}

\newcommand{\bbF}{\FF}

\newcommand{\obar}[1]{\bar{#1}}
\newcommand{\ubar}[1]{\text{\underline{$#1$}}} 

\newcommand{\fP}{\mathfrak{P}}
\newcommand{\R}{\mathbb{R}}

\DeclareMathOperator{\sem}{sem}

\DeclareMathOperator{\Lip}{Lip}
\DeclareMathOperator{\ac}{ac} 

\newcommand{\C}{\mathbb{C}}

\newcommand{\cNN}{\mathscr{N\hspace{-2mm}N}}

\usepackage[markup=nocolor]{changes}
\usepackage{cancel}

\usepackage[ruled,vlined]{algorithm2e}

\begin{document}

\title[Affine models with path dependence]{Affine models with path-dependence under parameter uncertainty and their application in Finance}

\author[Geuchen]{Benedikt Geuchen}
\address{56 Rue de la Rochette, 77000 Melun, France}
\email{b.geuchen@web.de}
\author[Oberpriller]{Katharina Oberpriller}
        \address{University of Munich, Theresienstr. 39, 8033 M\"unchen, Germany.}
        \email{oberpriller@math.lmu.de}
\author[Schmidt]{Thorsten Schmidt}
        \address{Albert-Ludwigs University of Freiburg, Ernst-Zermelo-Str. 1, 79104 Freiburg, Germany.}
        \email{ thorsten.schmidt@stochastik.uni-freiburg.de}

    \date{\today. }
    \thanks{
    The authors acknowledge support by the state of Baden-W\"urttemberg through bwHPC and the German Research Foundation (DFG) through grant INST 35/1134-1 FUGG and through the DFG grant SCHM 2160/15-1}

\begin{abstract}
In this work we consider one-dimensional generalized affine processes under the paradigm of Knightian uncertainty (so-called non-linear generalized affine models). This extends and generalizes previous results in \cite{FadinaNeufeldSchmidt2019} and \cite{lutkebohmert2021robust}. In particular, we study the case when the payoff is allowed to depend on the path, like it is the case for barrier options or Asian options.

{To} this end, we  develop the path-dependent setting for the value function 
relying on functional It\^o calculus. We establish a dynamic programming principle which then leads to a functional non-linear Kolmogorov equation describing the evolution of the value function. 
While for Asian options, the valuation can be traced back to PDE methods, this is no longer possible for more complicated payoffs like barrier options. To handle such payoffs in an efficient manner, we approximate the functional derivatives with deep neural networks and show that the numerical valuation under parameter uncertainty is highly tractable. 

Finally, we consider the application to structural modelling of credit and  counterparty risk, where both parameter uncertainty and path-dependence are crucial and the approach proposed here opens the door to efficient numerical methods in this field. 
\bigskip

\end{abstract}

\maketitle 

\section{Introduction}

Knightian uncertainty plays an increasingly important role in modern financial mathematics.
This is because  \emph{model risk}
is always a concern in rapidly changing dynamic environments such as financial markets. In the approach taken in this work we overcome model risk by not  fixing a single (and of course
in reality unknown) probabilistic model, and consider instead a suitable class of  models. Frank Knight proposed to consider  worst-case prices over this class as a natural approach which acknowledges uncertainty. For practical applications, it turns out that (too) large classes of models lead to prohibitively expensive worst-case prices. The goal of our work is  therefore to study a suitably large and flexible class that does not result in unrealistic prices. 

The class we will focus on, \emph{generalized non-linear affine models} as proposed in \cite{lutkebohmert2021robust}, extends the class of affine models  and provides a suitable representation of uncertainty. 
Affine models are a frequently used model class in practice since they  combine a large flexibility together with a high
tractability (see for example \cite{DuffieFilipovicSchachermayer}, \cite{KellerResselSchmidtWardenga2018} and references therein). 

Treating uncertainty in dynamic frameworks is a delicate question, since already establishing a dynamic programming principle requires checking measurability and pasting properties. {The pioneering papers dealing with this topic are e.g. \cite{denis_hu_peng}, \cite{NeufeldNutz2014}, \cite{NeufeldNutz2017}, \cite{Nutz2013}, \cite{peng2019nonlinear} and \cite{peng}. Further contributions relying on second-order backward stochastic differential equations have been provided in \cite{cheridito-et-al-05}, \cite{matoussi_possamai_zhou}, \cite{soner_touzi_zhang_2012} and \cite{soner_touzi_zhang_2013}.} The work \cite{FadinaNeufeldSchmidt2019} was, to the best of our knowledge, the first one which considered uncertainty in an affine setting. In contrast to $G$-Brownian motion or $G$-L\'evy processes, studied for example in \cite{peng2019nonlinear}, \cite{NeufeldNutz2017}, \cite{denk2017semigroup}, the set of probability measures which are taken into account in this class depends on the state of the observation. This is a key feature in financial markets: uncertainty for a stock being worth a few pennies naturally differs from a stock being worth 1.000\$. 

Taking Knightian uncertainty into account results in pricing mechanisms which are robust against model risk and hence play a prominent role in the literature, see \cite{Martini,cont-06,Acciaioetal2016,muhle2018risk,bielecki2018arbitrage}, and the book \cite{guyon2013nonlinear}. {Moreover, a link to dual pricing has been established for example in \cite{bartl_kupper_neufeld}, \cite{dolinsky_soner}, \cite{HERNANDEZHERNANDEZ2007980}, \cite{nutz_2015}, \cite{park_wong_2022}, and \cite{schied_2007}.} The class of \emph{non-linear generalized affine processes} (i.e.\ generalized affine processes under parameter uncertainty) are able to capture this phenomena. Non-linear affine processes have been applied to a reduced-form intensity model in \cite{biagini2020reduced} and have been further extended in \cite{biagini_bollweg_oberpriller_2022}, \cite{criens2022non}. 

In this work we extend the non-linear generalized affine model to the path-dependent case. This requires a technical extension of the setup using functional It\^o calculus which has been established in \cite{ContFournie2013}. As in \cite{FadinaNeufeldSchmidt2019} and \cite{lutkebohmert2021robust}, we associate non-linear expectations to this class of path-dependent non-linear generalized affine processes and derive the corresponding dynamic programming principle. Furthermore, we provide a functional version of the non-linear Kolmogorov equation. The latter result opens the door to the robust valuation of path-dependent options. Note that the theoretical results in \cite{FadinaNeufeldSchmidt2019} and \cite{lutkebohmert2021robust} only allow the robust pricing of derivatives whose payoff depends on the value of the underlying at the time of maturity. Moreover, here the valuation relies on a non-linear PDE while in the path-dependent case we work with a non-linear path-dependent PDE.

{While in  special cases the value function turns out to be the solution of a linear PDE (see Proposition \ref{korollar:kolmogorov_folgenrung2}), we generally have to solve a non-linear functional Kolmogorov equation in order to valuate path-dependent options. Dealing with such an equation numerically leads to  the following two problems. First, the non-linear structure of the PDE raises some non-trivial difficulties. Second, the path-dependence of the PDE adds an extra layer of complexity. Thus, it is not possible to come up with a convincing numerical approach for this kind of highly complicated equations in full generality. The aim of the examples we provide in this paper is to outline conditions under which a tractable valuation is feasible.}

{To do so, we first focus on Asian options by making use of the smoothing effect of the average value such options depend on. Second, we consider more general payoffs and rely on machine learning methods for the evaluation.} More specifically, we rewrite the problem in terms of a forward-backward SDE and then use deep neural networks to approximate the functional derivatives, as suggested in \cite{EHanJentzen2017,BeckEJentzen2019}. We illustrate this approach for an up-and in digital option and conclude by comparing our results to the existing methods.

To the best of our knowledge these two approaches are the first results which allow to make use of non-linear affine processes with path-dependence in applications and thus could be a relevant starting point for further research into this area. At the same time we would like to stress that the necessary conditions, which allow us to gain some more regularity of the associated value function, are an additional technical challenge for applications which needs to be studied further in future works.

The paper is organized as follows: in Section \ref{sec:UncertaintyForStochasticProcesses} we introduce path-dependent non-linear affine processes. In Section \ref{sec:DynamicProgramming} we derive the dynamic programming principle and in Section \ref{subsec:kolmogorovgleichung} we provide a functional non-linear Kolmogorov equation. In Section \ref{subsec:beispiele} we study the valuation of Asian options, whereas we focus on barrier options and the numerical solution of functional PDEs in Section \ref{sec:BarierOptions}.

\section{Path-dependent uncertainty for stochastic processes} \label{sec:UncertaintyForStochasticProcesses}

The class of (generalized) affine processes is a well-established class in financial applications and beyond. One reason for this is its high tractability due to an expression of the Fourier transform in terms of Riccati equations. {The successful application of these models in a dynamic environment like financial markets requires an appropriate treatment of the associated model risk. Our aim is there to introduce a suitable toolkit for this, using Knightian uncertainty, even if the high tractability of classical affine models is lost.} In particular, for settings with path-dependence, high tractability is rather exceptional. Nevertheless, we provide a non-linear Kolmogorov equation which opens the door to fast machine learning methods. 

Before we pursue this, we introduce a precise formulation of the  setting we are working in. To {do so, we} fix a  time horizon $T>0$ and let $\Omega =  C([0,T])$ be the canonical space of continuous, one-dimensional paths until $T$. 
 We endow $\Omega$ with the topology of uniform convergence and denote by $\ccF$ its Borel $\sigma$-field. 
Let $X$ be the canonical process $X(t,\omega) = \omega(t)$, 
 and let  $\mathbb{F}= (\ccF_t )_{t \geq 0}$ with $\ccF_t = \sigma (X(s), 0 \leq s \leq t)$ be the (raw) filtration generated by $X$.

\begin{remark}[On the multi{-}dimensional setting] \label{rem:multidim}
While in the core of the paper we study the one-dimensional framework and therefore concentrate on one-dimensional processes, the extension to higher dimensions is treated in Section \ref{subsec:Counterparty_risk}. It is a straightforward {generalization} and to keep notation simple we postpone this to a later point. 	
\end{remark}

Since we are interested in path-dependency, we need a proper notation to distinguish between the value of a process $X$ at time $t$ and the path until time $t$. For a path, or respectively a stochastic process, $x$ in the space $C([0,T])$ we denote by $x(t)$ its value at time $t \in [0,T],$ while $x_t=(x(s): 0 \le s \le t)$  denotes its path from $0$ to $t$. {This notation will be used throughout the paper.}

Next, we denote by $\fP (\Omega)$  the Polish space of all probability measures on $\Omega$ equipped with the topology of weak convergence\footnote{The weak topology is the topology induced by the bounded continuous functions on $\Omega$. Then, $\fP(\Omega)$ is a separable metric space and we denote the associated Borel $\sigma$-field by $\ccB(\fP(\Omega))$. }. 
{For} a probability measure  $P \in \mathfrak{P}(\Omega)$ and $t \in [0,T]$, we say that the process $X$ from time $t$ on, denoted by  $X^t=X^t(s)_{s \in [0,T-t]}=(X(t+s))_{s \in [0,T-t]}$, is a  $P$-semimartingale for its raw right-continuous version of the filtration, if there exists a (continuous)  $P$-local martingale $M^{t,P}=(M^{t,P}(s))_{s \in [0,T-t]}=(M^P(t+s))_{s \in [0,T-t]}$  and a continuous, adapted process  $B^{t,P} = (B^{t,P}(s))_{s \in [0,T-t]}=(B^P(t+s))_{s \in [0,T-t]}$ which has $P$-a.s.\ paths with finite variation, such that $M^{t,P}(0) = B^{t,P}(0) = 0$ $P$-a.s. and
\begin{align*}
X^t = X(t) + M^{t,P} + B^{t,P} \quad P\text{-a.s.}
\end{align*}

A semimartingale can be described by its \emph{semimartingale characteristics}: we denote by $C^t=(C^t(s))_{s \in [0,T-t]}=(C(t+s))_{s \in [0,T-t]}$ the quadratic variation of $X^t$ (which can be chosen independently of $P$). 
Hence, $C^t(s) = \langle X^t \rangle(s) = \langle M^{t,P} \rangle(s)$, for $s \in [0,T-t]$, where $\langle X \rangle$ denotes the predictable quadratic variation of $X$. The pair $(B^{t,P},C^t)$ is called \emph{(semimartingale) characteristics} of $X^t$. 

The class of processes we are aiming at, are  semimartingales with absolutely continuous characteristics. A semimartingale $X^t$ has absolutely continuous characteristics, if its semimartingale characteristics satisfy
\begin{align*}
B^{t,P}(s) = \int_0^{s} \beta^{t,P}(u) du, \quad C^t(s) = \int_0^{s} \alpha^t(u) du, \quad \forall s \in [0,T-t], \ P\text{-a.s.}
\end{align*}
with predictable processes $\beta^{t,P}=(\beta^{t,P}(s))_{s \in [0,T-t] }=(\beta^P(t+s))_{s \in [0,T-t] }$ and $\alpha^t=(\alpha^t(s))_{s \in [0,T-t] }=(\alpha(t+s))_{s \in [0,T-t] }$. The pair $(\beta^P,\alpha)$ is called \emph{differential characteristics} of $X$. We note that $\alpha$ can be chosen independently of $P$ since the quadratic variation is a path property. Moreover, {the differential characteristics} are independent of $t$ in the following sense: if $X$ is a semimartingale on $[t_1,T]$ and $t_2 > t_1$, then {$X$} is also a semimartingale on $[t_2,T]$ {and} the differential characteristics coincide on $[t_2,T]$. 

For $t \in [0,T]$ we denote by $$\psemac(t)$$ the set of probability measures such that $X^t$ is a $P$-semimartingale for its raw right-continuous filtration, which has absolutely continuous characteristics. {By Proposition 2.2 in \cite{NeufeldNutz2014}, $X^t$ is a semimartingale with respect to the raw filtration if and only if it is a semimartingale with respect to its right-continuous {extension}, such that it is no restriction to consider the {raw} right-continuous filtration here.}

\bigskip

\subsection{Generalized affine processes under uncertainty}

In {the classical setting, a one-dimensional} generalized affine process (without parameter uncertainty) is the unique strong solution of the stochastic differential equation
\begin{align}
\label{eq:affiner_prozess_sde}
dX(t) = \left(b_0+b_1X(t)\right)dt + \left(a_0+a_1X(t)\right)^{\gamma} \,dW(t), \quad X(0) = x_0,
\end{align}
where $a_i,b_i \in \mathbb{R}, \ i=0,1, \ \gamma \in [1/2,1]$ and $W$ is a Brownian motion. {This class of diffusion processes is very flexible and} as special cases it encompasses for example the Black-Scholes model, the Vasi\v cek model, the {Cox-Ingersoll-Ross} model, and the {Constant Elasticity of Variance} model. 

\smallskip 

The class of affine models and its generalizations has intensively been studied in the literature and we refer to \cite{DuffieFilipovicSchachermayer} or to \cite{cuchieroFilipovicTeichmann:ATS} for further information. In particular, they show excellent tractability, which allows for fast and flexible calibration, see (inter alia) \cite{EberleinGrbacSchmidt2013, gumbel2020machine, grbac2015affine}. Besides interest rates, applications in credit risk (\cite{FilipovicSchmidt2010,FilipovicSchmidtOverbeck2011,TSchmidt_InfiniteFactors, ErraisGieseckeGoldberg2010}, to name just a few) or stochastic mortality (see for example \cite{schrager2006affine, russo2011calibrating, biffis2005affine,zeddouk2020mean}) are very popular. 

\smallskip

In all {these} practical applications, {parameters of the underlying model} need to be estimated. The research in this field, and in particular the above listed references,  highlight that estimation and calibration is far from straightforward. Hence, a certain degree of model risk always remains. This motivates the  concise treatment of parameter uncertainty  in applications to avoid overlooking uncertainty in estimated or calibrated parameters. Since many derivatives and insurance contracts turn out to be path-dependent,
we will study this topic in detail. To the best of our knowledge this  is the first time that parameter uncertainty for path-dependent payoffs is considered for a general and yet tractable model class.

\smallskip

To this end, we consider \emph{parameter uncertainty}, represented by uncertainty in the differential characteristics of the observed process. We will detail this step precisely in the following. A key observation will be made: the developed class of \emph{non-linear generalized affine processes} (NGA) encompasses \emph{all} semimartingale laws whose differential characteristics remain in between certain boundaries. Thus, this class contains significantly more models than only the Markovian models listed above and hence itself becomes a highly flexible class. The reason why to restrict to NGA and not studying larger classes is our aim to provide rather small classes suitable for practical applications. Larger classes (like for example \emph{all} diffusions) typically contain many risks which an investor might classify as unlikely. This might result into unrealistic price mechanisms since many investors might not be willing to pay for these seemingly unlikely risks. 

\smallskip

The differential characteristics of the process $X$ in Equation \eqref{eq:affiner_prozess_sde} are given by $\beta = b_0+b_1X$ and $\alpha = (a_0+a_1X)^{2\gamma}$.
For $\gamma=1/2$, we recover the case of an  affine process.
\emph{Parameter uncertainty} is inspired by statistics or calibration exercises, where the parameter vector $\theta=(b_0,b_1,a_0,a_1, \gamma)$ is not known but needs to be estimated from data. This estimation comes with an error and we specify error bounds (for example obtained from confidence intervals of {the} calibration to bid- and ask-prices) 
 $\underline b_i \le \overline b_i$, $\underline a_i \le \bar a_i$, $i=0,1$, and $\underline \gamma \leq \overline \gamma$, respectively. 
 
 Summarizing, parameter uncertainty in the NGA is specified by the set 
\begin{align}
\label{Theta}
\Theta := 
[\underline b_0,\overline b_0] \times 
[\underline b_1,\overline b_1] \times 
[\underline a_0,\overline{{a_{0}}}] \times 
[\underline a_1, \overline a_1]  \times
[\underline \gamma, \overline \gamma]
\subset \R^2 \times \R_{\ge 0}^2 \times [1/2, 1]{,}
\end{align}
and the case {in which} $\Theta$ is a singleton correspsonds to the classical case without uncertainty.

To assess the semimartingale characteristics, we need to specify the intervals generated by the associated affine functions. In this regard, we {set}
$B:=[\underline b_0,\overline b_0] \times [\underline b_1,\overline b_{1}]$, 
$A:=[\underline a_0,\overline a_0] \times [\underline a_1,\overline a_{1}]$ and $\Gamma:=[\underline \gamma, \overline \gamma].$
As usual we denote by $(a_0,a_1)$ the two components of the vector $a \in \R^2$, and similarly for $b\in \R^2$. The uncertainty intervals around the current observation $\{x=X(t)\}$ are denoted by 
\begin{align} \label{def:AB}
    \begin{aligned} 
    b^*(x) & := \lbrace b_{0}+b_{1}x: b \in B \rbrace, \quad \text{ and } \quad 
    a^*(x) & := \lbrace \left (a_0+a_1x^+\right)^{2 \gamma}: a \in A, \gamma \in \Gamma  \rbrace,
\end{aligned} \end{align} 
for $x \in \R$.
As the state space will, in general, be $\R$, we have to ensure non-negativity of the quadratic variation  which is achieved using $(\cdot)^+:=\max\{\cdot,0\}$ in the definition of $a^*$.

We now introduce a generalized affine process under uncertainty which follows a specific path $x \in C({[}0,t{]})$ until time $t \in [0,T]$. 
\begin{definition}\label{def:nlaffine}
Fix $t \in [0,T]$ and $x \in C([0,t])$. 
A \emph{non-linear generalized affine process with history $x$ until $t$} is the family of semimartingale laws {$ \cA(t,x_t,\Theta)$, such that for each} $P \in \fP_{\sem}^{\ac}(t)$ {it holds}
\begin{enumerate}[(i)]
  \item  $P(X_t=x_t)=1$, 
  \item  $\displaystyle \beta^P(s) \in b^*(X(s)) $, and $\alpha(s) \in a^*(X(s)) $ for $dP\otimes dt$-almost all $(\omega,s) \in \Omega \times (t,T]$.
\end{enumerate}
For $\ubar{\gamma}=\obar{\gamma}=1/2$ we call such a family of semimartingale laws $P \in \fP_{\sem}^{\ac}(t)$ satisfying (i) and (ii) \emph{non-linear affine process with history $x$ until $t$.} 
\end{definition}

{Note that for $t=T,$ condition (ii) in the above definition is automatically satisfied  and thus $\cA(T,x_{{T}},\Theta)$ consists of all semimartingale measures $P \in \fP_{\sem}^{\ac}(T)$ such that $P(X_T=x_T)=1$. }

As explained in the introduction, parameter uncertainty is represented by a family of models replacing the single model in the approaches without uncertainty: according to Definition \ref{def:nlaffine}, the affine process under parameter uncertainty is represented by a \emph{family} of semimartingale laws instead of a single one. 
Intuitively, the two parts of the definition correspond to a generalized affine process under uncertainty (condition (ii)), which follow{s} the path $x_t=(x(s):s \in [0,t])$ until $t$ (condition (i)).

\begin{remark}[On the generality of the class {$\cA$}]
The class of non-linear generalized affine processes contains all Markov processes solving Equation \eqref{eq:affiner_prozess_sde}. In addition, it contains a large number of non-Markovian processes. {In particular}, all semimartingales  whose  characteristics remain in the set-valued processes $b^*(X)$ and $a^*(X)$ are included. As an example one could think of a SABR model conditioned on the restriction that the stochastic volatility $\sigma$ remains in the interval $[ \ubar a_1, \obar a_1]$. If {this} interval is not too small, {it is natural to assume that the volatility lies in this interval.}\end{remark}

\subsection{First properties}
By classical existence results for stochastic differential equations, we obtain that the class $\cA(t,x_t,\Theta)$ is not empty. 

\begin{proposition}
\label{prop:existenz_affiner_prozesse}
Consider  $t \in [0,T]$ and measurable mappings $b_0,b_1,a_0,a_1, \gamma : [t,T] \times \R \to \R$,  each with values in  $[\ubar{b}_0,\obar{b}_0]$, $[\ubar{b}_1,\obar{b}_1]$, $[\ubar{a}_0,\obar{a}_0]$, $[\ubar{a}_1,\obar{a}_1]$ and $[\ubar{\gamma}, \obar{\gamma}]$, respectively. 
Define
\begin{align}
b(s,y) 
&:= 
b_0(s,y) + b_1(s,y)y, \label{eq:SetNonEmpty1}\\ 
a(s,y) & := (a_0(s,y) + a_1(s,y)y^+)^{2\gamma(s,y)}\label{eq:SetNonEmpty2}
\end{align}
and assume that  $b(s,\cdot)$ and $a(s,\cdot)$ are continuous for all $s \in [t,T]$.
Then, for all $x \in C([0,T])$ there exists a $P \in \aset$, such that for the differential characteristics $(\beta^P,\alpha)$ under $P$,
\begin{align*}
\beta^P (\omega,s) = b(s,X(\omega,s)) \quad \text{ and } \quad \alpha(\omega,s) = a(s,X(\omega,s)) 
\end{align*}
for $dP \otimes dt$-almost all $(\omega,s) \in \Omega \times (t,T]$.
\end{proposition}

\begin{proof}
Observe that by assumption $b$ and $a$ are jointly continuous and have at most linear growth in the sense of Equation III.2.4, p. 258 in \cite{GikhmanSkorokhoddIII}. Then existence of a weak solution follows by Skorokhod's classical existence result, see Theorem III.2.4, p. 265 in \cite{GikhmanSkorokhoddIII}.

\end{proof}

We denote by $\cX \subseteq \R$ the state space which we will consider. We either will be interested in the general case, where Gaussian processes are allowed (for example the Vasi\v cek model) and hence, the state space will be $\R$. Note that this does not mean that a CIR model will be excluded, see Example \ref{ex:VCIR}. Or, we will focus on non-negative values only, for example when considering stock prices and we will consider $\R_{>0}$ as state space.

More precisely, we focus on the cases where either
\begin{align}
\label{eq:parametermenge}
\begin{split}
\text{(i)}\ &\ubar{a}_0 > 0 \text{ and } \cX = \R,\\
\text{(ii)}\ & \ubar{\gamma}=\obar{\gamma}=1/2,\ \ubar{a}_0 = \obar{a}_0 = 0,\ \ubar{b}_0 \geq \frac{\obar{a}_1}{2} > 0 \text{ and } \cX = \R_{>0}, \text{ or } \\
\text{(iii)}\ & \frac{1}{2}< \ubar{\gamma} \leq \obar{\gamma} \leq 1,\ \ubar{a}_0 = \obar{a}_0 = 0,\ \ubar{b}_0  > 0, \ \ubar{a}_1>0  \text{ and } \cX = \R_{>0}.
\end{split}
\end{align}
In this regard, we call $\Theta$ \emph{proper} if either (i), (ii) or (iii) is satisfied. In these cases, 
\begin{align*}
P(X_T \in C([0,T],\cX)) = P(C([0,T],\cX)) = 1 
\end{align*}
for all $P \in \cA(t,x_t,\Theta)$ and for all $(t,x) \in [0,T] \times C([0,T],\cX)$.
This trivially holds true in the case (i) and in the case (ii) -- see \cite[Proposition 2.3]{FadinaNeufeldSchmidt2019} for a proof of the latter claim. The case (iii)  was proven in \cite[Lemma A.1]{lutkebohmert2021robust}.

The following example illustrates the additional flexibility gained by incorporating parameter uncertainty. 
\begin{example}[The Vasi\v cek-CIR model]\label{ex:VCIR}
It is surprising that in the non-linear setting one does not need to choose between a Vasi\v cek and a CIR model. Indeed, consider the case (i) of Equation \eqref{eq:parametermenge} and assume $\ubar a_0 > 0$ and $\cX = \R$. Moreover, we assume also that $\ubar a_1>0$ and hence this model is able to interpolate between Vasi\v cek and CIR. Depending on $\{X  {(t)} <0\}$ or $\{X{(t)} > 0\}$, the model behaves more like the earlier or the latter. For practical purposes one typically would start from confidence intervals on the mean-reversion speed and mean-reversion level and obtain respective intervals for $b_0$ and $b_1$. Note that it is also possible to consider time-dependent parameters (a Hull-White extension). \end{example}

\section{Dynamic Programming} \label{sec:DynamicProgramming}

In this section we consider the valuation problem and establish a dynamic programming principle relying on Theorem 2.1 in \cite{ElKarouiTan2015}. We consider a payoff which is given as a measurable, bounded and possibly path-dependent function 
\begin{align} \label{eq:PayoffFunction}
\Phi : C([0,T],\cX) \to \R.
\end{align}
{Unbounded payoffs can be  approximated by choosing a suitable large bound. Together with subadditivity and some regularity of the payoff function this allows to obtain an arbitrarily close approximation. 
}

In the paradigm of Knightian uncertainty, the value {function} at time $t$, given the observed path $x_t$ until $t$, is specified in terms of the supremum over all possible expectations {by}, 
{\begin{align}\label{eq:V}
V_t: C([0,t],\cX) \to \R, \qquad V_t(x_t) := \sup_{P \in \mathcal{A}(t,x_t,\Theta)}E_P\left[\Phi(X_T)\right].
\end{align}
Moreover, we define the function $V:[0,T] \times C([0,T],\cX) \to \R$ by
\begin{align} \label{eq:V1}
	V(t,x):=V_t(x_t),
\end{align}
where $x_t$ is understood as the restriction of the path ${x}$ to $[0,t]$.
}

{In the remaining part we will work  under the following technical assumption. 
\begin{assumption} \label{AssumptionStanding}
	Let the set $\Theta$ in \eqref{Theta} be proper and the payoff function $\Phi$ in \eqref{eq:PayoffFunction} be measurable and bounded. 
\end{assumption}}

Before we prove the dynamic programming principle, we first establish measurability of the value function. 

\begin{proposition}\label{prop:DPPa}
    {Let Assumption \ref{AssumptionStanding} hold}. Then, the value function $V:[0,T]\times C([0,T],\cX) \to \R$ is upper semianalytic.
\end{proposition}

\begin{proof}
We consider the Polish spaces $E = [0,T] \times \Omega$ and $F = \mathfrak{P}(\Omega)$. By similar arguments as in Section 1 in \cite{criens2022non}, it can be shown that
\begin{align*}
\ccA := \{(P,t,x) \in \mathfrak{P}(\Omega) \times [0,T] \times \Omega : P \in \cA(t,x_t,\Theta)\} \subseteq F \times E
\end{align*}
is measurable and hence analytic. It follows now from \cite[Prop{osition} 7.25]{BertsekasShreve1978} with monotone class arguments (see also \cite[Lemma 3.1]{NeufeldNutz2014}), that 
\begin{align*}
    f: (P,t,x) \mapsto E_P[ \Phi(X_T) ] 
\end{align*}
is measurable. The restriction to $\ccA$ is upper semianalytic since  $\{f > c\} \cap \ccA  \in \ccB(F \times E)$. 

The next step is to apply \cite[Prop{osition} 7.47]{BertsekasShreve1978}. In this regard, we compute 
\begin{align*}
\proj_E(\ccA) &= \{(t,x) \in E : \cA(t,x_t,\Theta) \ne \varnothing\}
= E
\end{align*}
by Proposition \ref{prop:existenz_affiner_prozesse}. The mentioned result gives that the function
\begin{align*}
f^*: (t,x)\mapsto \sup_{P : (P,t,x) \in \ccA}f(P,t,x) 
\end{align*}
is upper semianalytic and so is its restriction to $[0,T] \times C([0,T],\cX)$. 
\end{proof}

Finally, we have all ingredients to prove the \emph{dynamic programming principle}.

\begin{theorem}
\label{theorem:dynamic_programming}
{Let Assumption \ref{AssumptionStanding} hold}.
For all $t \in [0,T]$, $x \in C([0,T],\cX)$  and every stopping time $\tau \subseteq [t,T]$,
\begin{align}\label{eq:DPP}
V(t,x) = \sup_{P \in \cA(t,x_t,\Theta)}E_P\left[V(\tau,X) \right].
\end{align}
\end{theorem}
{\begin{proof}For proving the dynamic programming principle we apply Theorem 2.1 in \cite{ElKarouiTan2015}. This can be directly done by extending the arguments in {Theorem 2.1 in }\cite{NeufeldNutz2017} or in {Theorem 3.1 in }\cite{criens2022non} to our setting. Since the core of this paper is on presenting the functional Kolmogorov equation, we skip a detailed presentation of the arguments.\end{proof}}

\section{The non-linear Kolmogorov equation}
\label{subsec:kolmogorovgleichung}

The main result of this paper is  a functional version of the non-linear Kolmogorov equation. This extends \cite{FadinaNeufeldSchmidt2019} to the setting with path-dependence. Note that the non-linear conditional expectation given in Equation \eqref{eq:DPP} is highly intractable -- it is not easily reachable by Monte-Carlo simulation since one has to simulate over a variety of models. The non-linear Kolmogorov equation is therefore the central tool for the relatively high degree of numerical tractability of generalized affine models under uncertainty.  

Mutatis mutandis we obtain the following lemma from \cite[Lemma A.{4}]{lutkebohmert2021robust}. 
\begin{lemma}
\label{lemma:supremum_abschÃ¤tzung}
Consider proper $\Theta$ and $q \ge 1$. Then, there exists  ${0 < \epsilon = \epsilon(q) < 1}$, such that for all  $0 < h \leq \epsilon$, all $t \in [0,T-h]$ and all $x \in C([0,T],\cX)$ it holds that
\begin{align*}
\sup_{P \in \aset} E_P \left[ \sup_{0 \leq s \leq h} |X(t+s) - x(t)|^q \right] \leq C(h^q+h^{q/2})
\end{align*}
with a constant 
$C = C(x(t),q) > 0$, being independent of $h$ and depending on $t$ only via $x(t)$. 
\end{lemma}

\subsection{Functional derivatives} \label{subsection:FunctionalDerivatives}
In the following, we introduce some notations and definitions about functional derivatives which we use later. For more details on this topic we refer to \cite{ContFournie2013}. {From now on, $D([0,T],\cX)$ denotes the space of c\`{a}dl\`{a}g functions from $[0,T]$ to $\cX$. On $[0,T] \times \cX$ we define the distance $d$ by
\begin{align} \label{eq:PseudoNorm}
d((t,x), (t',x')):=\vert t-t' \vert + \sup_{u \in [0,T]} \vert x_{t}(u) - x'_{t'}(u) \vert.
\end{align}}

For a path $x \in D([0,T],\cX)$, $0\leq t<T$ and $0\leq  h \leq T-t $ we denote by $x_{t,h}\in D([0,t+h],\cX)$ the \emph{horizontal extension} of $x_t$ to $[0,t+h]$, i.e. 
\begin{equation*}
    x_{t,h}(s):=x(s) \Ind_{[0,t]}(s)+ x(t) \Ind_{(t,t+h]}(s).
\end{equation*} 
Moreover, for $h \in \cX$ the \emph{vertical perturbation}  $x_{t}^h$ of $x_t$ is obtained by shifting the endpoint $x(t)$ with $h$, i.e.
\begin{equation*}
    x_{t}^h(s):=x_t(s) \Ind_{[0,t)}(s)+ (x(t)+h) \Ind_{\lbrace t \rbrace}(s).
\end{equation*}

The \emph{horizontal derivative} at $(t,x) \in [0,T) \times D([0,T],\cX)$ of the non-anticipative functional $F$ 
is defined as
$$
\cD_tF(x_t):= \lim_{h \searrow 0} \frac1h (F_{t+h}(x_{t,h}) - F_t(x_t)),
$$
if the limit exists. Moreover, we define 
the \emph{vertical derivative} at $(t,x) \in [0,T] \times D([0,T],\cX)$ of $F$  as
$$
\nabla_xF_t(x_t):= \lim_{h \searrow 0} \frac1h (F_{t}(x_{t}^{h}) - F_t(x_t)).
$$
\subsubsection*{Function classes} Next, we introduce some useful classes which we will utilize later on. 
First, we define the class of \emph{left-continuous} non-anticipative functionals $\mathbb{F}_l^{\infty}$, i.e.\ those functionals $F=(F_t)_{t \in [0,T)}$ which satisfy 
\begin{align*}
    \forall t \in [0,T], \ \forall \epsilon >0, \ \forall x \in D([0,t],\cX), \ \exists \eta >0, \ \forall h \in [0,t], \ \forall x' \in D([0,t-h],\cX) \\
    h + \sup_{s \in [0,t]} \vert x(s)-x'_{t-h,h}(s)\vert < \eta \quad \Rightarrow \quad  \vert F_t(x)-F_{t-h}(x')\vert < \epsilon.
\end{align*}
By $\mathbb{B}$ we denote the set of \emph{boundedness-preserving} functionals. These are those  non-anticipative functionals $F$ such that for any compact subset $K \subset \cX $ and any $R>0$, there exists a constant $C=C(K,R)>0$ such that
\begin{align*}
    \forall t \in [0,T], \ \forall x \in D([0,T],K): \sup_{s \in [0,t]} \vert x(s) \vert < R \quad \Rightarrow \quad \vert F_t(x) \vert < C.
\end{align*}

We denote by $\mathbb{C}_b^{1,2}([0,T))$ the class of non-anticipative functionals $F=(F_t)_{t \in [0,T]}$ for which 
\begin{enumerate}
    \item one horizontal derivative of $F$ exists at all $(t,x) \in [0,T) \times D([0,T],\cX)$,
    \item two vertical derivatives of $F$ exist at all $(t,x) \in [0,T) \times D([0,T],\cX)$,
    \item $F, \cD F, \nabla F, \nabla^2 F$ {are continuous at every point in $[0,T] \times D([0,T],\cX)$ with respect to the distance $d$ in \eqref{eq:PseudoNorm}},
    \item $F, \nabla F, \nabla^2 F$ are elements in $\mathbb{F}_l^{\infty}$.
     \item $ \cD F, \nabla F, \nabla^2 F$ are elements in $\mathbb{B}$.
\end{enumerate}

A non-anticipative functional  $F = (F_t)_{t \in [0,T)}$ is called \emph{Lipschitz continuous}, 
    if for all  $x,y \in D([0,T],\cX)$ and $t \in [0,T)$ there exists a constant $L > 0$  such that for all  $h \in [0,T-t)$,
\begin{align}
|F_{t+h}(x_{t+h}) - F_t(y_t)| \leq L \Big(h + \sup_{s \in [0,h]}|x(t+s)-y({t})|\Big).    \label{def:pfadweise_lipschitzstetig}
\end{align}

We denote by $\Lip_b^{1,2}([0,T))$ the class of non-anticipative functionals $F \in \mathbb{C}_b^{1,2}([0,T))$, whose derivatives are Lipschitz continuous in the sense of \eqref{def:pfadweise_lipschitzstetig} such that $\nabla_x F$ is uniformly bounded. It is also possible to let the Lipschitz constant depend on $x_t,y_t$.

\subsection{Viscosity solutions of PPDE}
In this section we consider viscosity solutions of the path-dependent Kolmogorov equation. For theoretical background, we refer to Section 11 of \cite{zhang2017backward}.
{Now we introduce the path-dependent \emph{non-linear Kolmogorov equation}, we need to study. This equation will specify the evolution of the value functions and in the case without uncertainty turns out to be the classical Kolmogorov equation. To this end, define
\begin{align} \label{eq:DefinitionNonlinearity}
G(y,p,q) := \sup_{(b_0,b_1,a_0,a_1, \gamma) \in \Theta} \Big((b_0 + b_1y)p + \Big(\frac12 (a_0 + a_1y^+)\Big)^{2\gamma}q \Big).
\end{align}
Consider a non-anticipative functional  $F = (F_t)_{t \in [0,T]}$, i.e.  $F(x)=F_t(x_{\cdot \wedge t})$ for all $x \in D([0,T],\cX)$ and $t \in [0,T)$,
and \begin{align}
\label{eq:kolmogorov}
\begin{split}
-\cD_tF(x_t) - G\left(x(t),\nabla_xF_t(x_t),\nabla_x^2F_t(x_t)\right) &= 0 \quad \text{on}\ [0,T) \times C([0,T],\cX),\\
F_T(x) &= \Phi(x), \quad  x \in \cX.
\end{split}
\end{align}}
 We choose the set of test functions to be properly differentiable, Lipschitz continuous functions, i.e. $\varphi \in \Lip_b^{1,2}([0,T))$.
This weakens the stronger requirement of differentiability ($C^{2,3}$) in \cite{FadinaNeufeldSchmidt2019} and \cite{lutkebohmert2021robust}.

Consider $F = (F_t)_{t \in [0,T]}$ with $F_t : C([0,t],\cX) \to \R$ for all $t \in [0,T]$. Then $F$ is called \emph{{weak sense} viscosity subsolution} of the Kolmogorov equation \eqref{eq:kolmogorov}, if $F_T(\cdot) \leq \Phi(\cdot)$, and for all  $(t,x) \in [0,T) \times C([0,T],\cX)$ and all  $\varphi \in \Lip_b^{1,2}([0,T))$ with  $\varphi_t(x_t) = F_t(x_t)$  and $\varphi_s(y_s) \geq F_s(y_s)$,  $(s,y) \in [0,T) \times C([0,T],\cX)$, it holds that:
\begin{align*}
-\cD_t\varphi(x_t) - G\left(x(t),\nabla_x\varphi_t(x_t),\nabla_x^2\varphi_t(x_t)\right) \leq 0.
\end{align*}
The definition of a \emph{{weak sense} viscosity supersolution} is obtained by reversing the inequalities. $F$ is called \emph{{weak sense} viscosity solution} if it is both a viscosity super- and a viscosity subsolution. 
{
\begin{remark}[Continuity and Uniqueness]
	Note that the solution $F$ in	
	\eqref{eq:kolmogorov} depends on the time and on the path. For such solutions one considers weak sense viscosity solutions, see e.g. \cite{criens2022non}, which a priori do not need to be continuous. Corollary 4.10 in \cite{criens2022non} provides sufficient conditions for the continuity of the solution of a path-dependent PDE, which are however not satisfied in our setting except for the L\'evy case, i.e. $\underline{a}_1=\overline{a}^1=0, \ \underline{b}_1=\overline{b}^1=0$  and $\underline{\gamma}=\overline{\gamma}=1/2$.
	
	Uniqueness of  the solution of \eqref{eq:kolmogorov} is more subtle. 
	To the best of our knowledge,  the only available results are very recent, dealing with path-dependent Hamilton-Jacobi-Bellman equations, see \cite{cosso_gozzi_rosestolatio_russo_2023}, \cite{zhou_2023} and \cite{criens2022non}. These results require additional global Lischpitz continuity on the coefficients and on the boundary condition and deliver uniqueness among all solutions with another certain Lipschitz  continuity.  
\end{remark}
}
\begin{theorem}
\label{thm:kolmogorov}
{If Assumption \ref{AssumptionStanding} holds,} the value function $V$ from Equation \eqref{eq:V} is a {weak sense} viscosity solution of the Kolmogorov equation \eqref{eq:kolmogorov}.
\end{theorem}
\begin{proof}
We remark that in the subsequent lines within this proof, $C > 0$ is a constant whose value may change from line to line. Moreover, we define the  constant
\begin{align}
\label{eq:parametermenge_konstante}
\cK := |\ubar{b}^0| + |\obar{b}^0| + |\ubar{b}^1| + |\obar{b}^1| + \obar{a}^0 + \obar{a}^1 + 1. 
\end{align}

\emph{{Weak sense} subviscosity:} consider $(t,x) \in [0,T) \times C([0,T],\cX)$ and $\varphi \in \Lip_b^{{1,2}}([0,T))$ with $\varphi_t(x_t) = V_t(x_t)$ and $\varphi_s(y_s) \geq V_s(y_s)$ for all $(s,y) \in [0,T) \times C([0,T],\cX)$. 
The dynamic programming principle, Theorem \ref{theorem:dynamic_programming}, yields that for $u \in (0,T-t)$, 
\begin{align*}
0 = \sup_{P \in \cA(t,x_t,\Theta)} E_P[V_{t+u}(X_{t+u}) - V_t(x_t)] \leq \sup_{P \in \cA(t,x_t,\Theta)}E_P[\varphi_{t+u}(X_{t+u}) - \varphi_t(x_t)].
\end{align*}
The expectation is well-defined since with $\Phi$ also $V$ is bounded and $\varphi \in \Lip_b^{{1,2}}([0,T))$.

Fix any  $P \in \aset$ and denote as above by $(\beta^P,\alpha)$ the differential characteristics of $X$ under $P$. 
By the functional It\^o-formula, see \cite[Th{eorem} 4.1] {ContFournie2013}, we obtain 
\begin{align}
\label{eq:kolmogorov_0}
\varphi_{t+u}(X_{t+u}) - \varphi_t(X_t)
    &= \int_0^u \cD_t \varphi_{t+s} (X_{t+s}) ds + \int_0^u \nabla_x \varphi_{t+s}(X_{t+s}) dM(t+s) \notag\\
    & \quad + \int_0^u \nabla_x \varphi_{t+s}(X_{t+s}) \beta^P(t+s) ds + \frac12 \int_0^u \nabla_x^2 \varphi_{t+s}(X_{t+s}) \alpha(t+s) ds,
\end{align}
$P$-a.s., where the stochastic integral is w.r.t.\ $P$ and $M$ is the martingale part in the $P$-semimartingale decomposition of $X$. 
Since $\nabla_x \varphi$ is uniformly bounded, from Remark {1} in \cite{FadinaNeufeldSchmidt2019},  for small 
enough $0 < u < T - t$, the local martingale part in \eqref{eq:kolmogorov_0} is in fact a true martingale and hence its expectation vanishes. 
Next, we consider the third addend in \eqref{eq:kolmogorov_0}. Then for all $s \in [0,u]$ we have
\begin{align}
\lefteqn{E_P \left[ \int_0^u \nabla_x \varphi_{t+s}(X_{t+s})\beta^P(t+s)ds \right]} \hspace{2cm}\notag\\
&\leq \int_0^u E_P \left[ |\nabla_x \varphi_{t+s}(X_{t+s}) - \nabla_x \varphi_t(x_t)||\beta^P(t+s)| + \nabla_x \varphi_t(x_t)\beta^P(t+s) \right]ds.\notag
\end{align}
By the Lipschitz continuity in \eqref{def:pfadweise_lipschitzstetig} for $\nabla_x \varphi$ and the definition of $\cK$ in \eqref{eq:parametermenge_konstante}, it follows for sufficiently small $u$
\begin{align}
    & \int_0^u E_P \left[ |\nabla_x \varphi_{t+s}(X_{t+s}) - \nabla_x \varphi_t(x_t)||\beta^P(t+s)|\right]ds \notag \\
    & \leq \int_0^u E_P \left[ L \Big(s + \sup_{v \in [0,u]}|X(t+v)-x(t)|\Big) \left(  \mathcal{K} + \mathcal{K}|x(t)| + \mathcal{K} \sup_{v \in [0,u]}|X(t+v) - x(t)|\right)\right]ds\label{eq:LemmaEstimate}  \\
    &\leq C \left( u^3 + u^{5/2} + u^2 + u^{3/2}\right){,} \notag
\end{align}
{where we use Lemma \ref{lemma:supremum_abschÃ¤tzung} in \eqref{eq:LemmaEstimate}.}
Thus, we get
\begin{align} \label{eq:KolmogorovEstimateDx}
{E_P \left[ \int_0^u \nabla_x \varphi_{t+s}(X_{t+s})\beta^P(t+s)ds \right]} &\leq \int_0^u E_P \left[ \nabla_x \varphi_t(x_t)\beta^P(t+s)  \right] ds  \notag \\ & \quad +  C \left( u^3 + u^{5/2} + u^2 + u^{3/2}\right).
\end{align}
By the same {arguments} we get the corresponding estimates for the first and the last addend in \eqref{eq:kolmogorov_0}. Hence, we have
\begin{align*}
    E_P \Big[ \varphi_{t+u} (X_{t+u})&-\varphi_t(X_t) \Big] \leq C \left( u^3 + u^{5/2} + u^2 + u^{3/2}\right) \notag  \\
    & + u \cD_t \varphi_t(x_t) +\int_0^{{u}} E_P \left[  G\left ( x(t+s), \nabla_x \varphi_{t+s}(x_{t+s}), \nabla_x^2 \varphi_{t+s}(x_{t+s})\right)\right]ds,
\end{align*}
where $G$ is defined in \eqref{eq:DefinitionNonlinearity}. The result follows  as in the proof of Theorem {1} of \cite{FadinaNeufeldSchmidt2019}.

\emph{{Weak sense} superviscosity:} consider  $(t,x) \in [0,T) \times C([0,T],\cX)$ and $\varphi \in \Lip_b^{{1,2}}([0,T))$ with $\varphi_t(x_t) = V_t(x_t)$ and  $\varphi_s(y_s) \leq V_s(y_s)$ for all  $(s,y) \in [0,T) \times C([0,T],\cX)$.
Then, for  $0 < u < T-t$,
\begin{align*}
0 = \sup_{P \in \cA(t,x_t,\Theta)} E_P[V_{t+u}(X_{t+u}) - V_t(x_t)] \geq \sup_{P \in \cA(t,x_t,\Theta)}E_P[\varphi_{t+u}(X_{t+u}) - \varphi_t(x_t)].
\end{align*}
Fix $P \in \aset$. As for \eqref{eq:KolmogorovEstimateDx}, we obtain 
\begin{align*}
E_P \left[ \int_0^u \nabla_x \varphi_{t+s}(X_{t+s})\beta^P(t+s)ds \right]
&\geq -C\left(u^3 + u^{5/2} + u^2 + u^{3/2}\right) + \int_0^u E_P\left[\nabla_x \varphi_t(x_t)\beta^P(t+s)\right]ds.
\end{align*}
Proceeding similarly for the other addends, we obtain {with  Lemma \ref{lemma:supremum_abschÃ¤tzung}} 
\begin{align*}
E_P[\varphi_{t+u}(X_{t+u}) - \varphi_t(x_t)]
&\geq -C\left(u^3 + u^{5/2} + u^2 + u^{3/2}\right) + u\cD_t \varphi(x_t)\notag\\
&\quad + \int_0^u E_P\left[\nabla_x \varphi_t(x_t)\beta^P(t+s) + \frac12 \nabla_x^2 \varphi_t(x_t)\alpha(t+s)\right]ds.
\end{align*}
Let $\theta=(b_0,b_1,a_0,a_1, \gamma) \in \Theta$. Then by Proposition \ref{prop:existenz_affiner_prozesse} there exists $P'=P(\theta)\in \aset$ such that 
\begin{align*}
\lefteqn{\int_0^u E_{P'}\left[\nabla_x \varphi_t(x_t)\beta^{P'}(t+s) + \frac12 \nabla_x^2 \varphi_t(x_t)\alpha(t+s)\right]ds} \hspace{2cm}\notag\\
&= \int_0^u E_{P'}\bigg[\nabla_x \varphi_t(x_t)\left(b_0(t+s,X(t+s))+b_1(t+s,X(t+s))X(t+s)\right)\notag\\
&\quad + \frac12 \nabla_x^2 \varphi_t(x_t)\left(a_0(t+s,X(t+s))+a_1(t+s,X(t+s))X(t+s)^+\right)^{2\gamma(t+s,X(t+s))}\bigg]ds\notag\\
&= \int_0^u E_{P'}\left[G_{\theta}\left(X(t+s),\nabla_x \varphi_t(x_t),\nabla_x^2 \varphi_t(x_t)\right)\right]ds.
\end{align*}
Here, we set
\begin{align*}
    G_{\theta}(y,p, q)&:=b(t,y)p+  \frac12  a(t,y) q,
\end{align*}
where the functions $a,b$ are as in \eqref{eq:SetNonEmpty1}, \eqref{eq:SetNonEmpty2}.
Hence, by a further application of {Lemma \ref{lemma:supremum_abschÃ¤tzung}},
\begin{align*}
\lefteqn{\int_0^u E_{P'}\left[G_{\theta}\left(X(t+s),\nabla_x \varphi_t(x_t),\nabla_x^2 \varphi_t(x_t)\right)\right]ds} \hspace{2cm} \notag\\
&\geq u\,G_{\theta}\left(x(t),\nabla_x \varphi_t(x_t),\nabla_x^2 \varphi_t(x_t)\right) - C \int_0^u E_{P'}\left[ \sup_{v \in [0,u]}|X(t+v)-x(t)| \right]ds\notag\\
&\geq u\,G_{\theta}\left(x(t),\nabla_x \varphi_t(x_t),\nabla_x^2 \varphi_t(x_t)\right) - C\left(u^2 + u^{3/2}\right).
\end{align*}
Summarizing, we established
\begin{align*}
&E_{P'}[\varphi_{t+u}(X_{t+u}) - \varphi_t(x_t)] \\
& \quad \geq -C\left(u^3 + u^{5/2} + u^2 + u^{3/2}\right) + u\cD_t \varphi(x_t) + uG_{\theta}\left(x(t),\nabla_x \varphi_t(x_t),\nabla_x^2 \varphi_t(x_t)\right).
\end{align*}
Dividing by $-u$ and letting $u \to 0$ it follows that 
\begin{align*}
-\cD_t \varphi(x_t) - G_{\theta}\left(x(t),\nabla_x \varphi_t(x_t),\nabla_x^2 \varphi_t(x_t)\right) &\geq 0.
\end{align*}
By taking the supremum over all $\theta \in \Theta$ it follows that $V$ is a viscosity solution. 
\end{proof}

{In the following we study the case {in which} the value function has a {more regular} solution. This question was already initiated in {Theorem 5 in} \cite{FadinaNeufeldSchmidt2019}, where it was shown that on a non-negative state space the worst-case with monotone value function turns out to be an affine process itself. 
Intuitively, in this case the supremum of Equation \eqref{eq:DefinitionNonlinearity} is always the upper bound and hence the problem becomes linear. We derive the associated consequences for the value function in the following result

Since in this case we will be able to exploit the underlying Markov structure, we introduce the family 
$$\cA(x): =\cA(0,x,\Theta), $$ 
for $x \in C(\lbrace{0\rbrace}, \cX)= \cX$. 
}

\begin{proposition}
    \label{korollar:kolmogorov_folgenrung2}
Consider the case  $\cX = \R_{>0}$ and $\ubar \gamma=\obar \gamma=\frac{1}{2}$ in Equation \eqref{eq:parametermenge}, i.e.\ $\ubar{a}^0 = \obar{a}^0 = 0,\ \ubar{b}^0 \geq \frac{\obar{a}^1}{2} > 0$. Assume that $\Phi$ is an increasing, convex function  which has at most polynomial growth. Then, the value function
\begin{align}\label{eq:DefineValueFunctionCorollary}
 {V(t,x):=\sup_{{P \in \cA(x)}} E_P[\Phi(X(T-t))], \quad (t,x) \in [0,T] \times \mathbb{R}_{>0}}
\end{align}
is a classical solution to the linear PDE
\begin{align}
    \label{eq:kolmogorv_special_case}
    \begin{cases}\partial_t V(t,x) + \big(\obar{b}^0 + \obar{b}^1x \big) \partial_x V(t,x) + \frac12 \obar{a}^1x\, \partial_{xx} V(t,x) = 0,& (t, x) \in [0, T) \times \R_{> 0},\\
    V (T,x) = \Phi (x),& x \in \R_{> 0}.
    \end{cases}
    \end{align}
Moreover, it is the unique solution among all solutions of polynomial growth.
    Furthermore, for every  ${x(0)} \in \R_{> 0}$, the  SDE 
    \begin{align}\label{eq:SDEFeller}
    dX(t) &= \big(\obar{b}^0 + \obar{b}^1X(t)\big)dt + \sqrt{\obar{a}^1\,X(t)}dW(t), \quad X (0) ={x(0)}, 
    \end{align}
    admits a {unique strong solution} with 
    \begin{align} \label{eq:ExpectationWorstCase}
    V(t,{x(0)}) = {E_{P}[\Phi(X(T - t))\, | \, X(0)=x(0)]}, \quad (t, x) \in [0, T] \times {\R_{>0}}.
    \end{align}
\end{proposition}

\begin{proof}{
If the state space $\cX=\R_{>0}$ consists only of positive numbers, monotonicity can be applied to simplify the approach. Indeed, recall equation \eqref{eq:DefinitionNonlinearity} and observe that under our assumptions it holds  for $p,q>0$ and $x>0$ that
\begin{align} G(x,p,q) = \sup_{(b_0,b_1,a_1, \gamma) \in \Theta} \Big((b_0 + b_1x)p + \frac12  (a_1x^+)q \Big) 
= (\bar b_0 + \bar b_1x)p + \frac12  \bar a_1x q. 
\end{align}
This is the generator of the Feller process or Cox-Ingersoll-Ross model in the SDE \eqref{eq:SDEFeller} - the worst case is reached by a classical affine model. 

To proceed further, we note that the SDE \eqref{eq:SDEFeller} has a unique strong solution with state space \(\mathbb{R}_{> 0}\) under the condition $ \ubar{b}^0 \geq \frac{\obar{a}^1}{2}>0$, see \cite[Section 6.3.1]{JeanblancChesneyYor2009}.
We denote its law by $P_{x(0)}$. This corresponds to an affine process, and hence is also a non-linear affine process such that  $P_{x(0)} \in \cA(x(0))$. Therefore, 
\begin{align}\label{temp:24}
\sup_{P \in \cA({x(0)})} E_P[\Phi(X(T-t))] \geq E_{P_{x(0)}}[\Phi(X(T-t))].
\end{align}
We also note that, similarly to \eqref{eq:moments}, all moments of the CIR process exists. Hence, 
$E_{P_{x(0)}}[\Phi(X(T-t))]>-\infty$ and all appearing expectations in \eqref{temp:24} are well-defined.

To achieve equality, we prove the other direction of this inequality relying on comparison results for general semimartingales. The details are shortly presented in the appendix.
We will apply Theorem 2.2 in \cite{Bergenthum_Rueschendorf_2007} to  $P \in\cA(x(0))$ and need to check the prerequisits: Note that by \ref{prop:PO} the propagation of order property for the increasing and convex function $\Phi$ follows. 
Moreover, for any semimartingale law $P \in \cA(x)$,  comparison of the differential characteristics holds in the following sense: if $(\beta,\alpha)$ denote the differential characteristics under $P$, then, again using $x \ge 0$,
$$ \beta(\omega,t) \le \obar{b}_0 + \obar{b}_1X(\omega,t), \qquad \text{and} \qquad  \alpha(\omega,t) \le  \obar{a}_1 X(\omega,t) $$
Since $\Phi$ is increasing and convex, Theorem 2.2 in \cite{Bergenthum_Rueschendorf_2007} yields that
$$
E_P[\Phi(X(T-t))] \leq E_{P_{x(0)}}[\Phi(X(T-t))], 
$$
and  equality 
\eqref{eq:ExpectationWorstCase} follows. }

It remains to show that for any $x:=x(0) \in \mathbb{R}_{>0}$ the value function $V$ in \eqref{eq:DefineValueFunctionCorollary} solves \eqref{eq:kolmogorv_special_case}. To do so we apply Theorem 6.1 in \cite{Janson_Tysk_2006} which allows to conclude that the stochastic representation $E_{P_{x}}[\Phi(X(T-t))]$ is a strong solution of the PDE in \eqref{eq:kolmogorv_special_case}, if $ (t, x) \mapsto E_{P_{x}}[\Phi(X(T-t))]$ is continuous. This continuity follows as in Theorem 2.15 in \cite{criens2022MarkovSelection}.
\end{proof}

{While in some special cases the value function is a solution of a linear PDE as in Proposition \ref{korollar:kolmogorov_folgenrung2}, we generally have to solve a non-linear functional Kolmogorov equation in order to valuate path-dependend options in a robust way, as Theorem \ref{thm:kolmogorov} shows. 
We already explained in the introduction that this leads to non-trivial difficulties, such that in general it is not possible to come up with a convincing numerical approach for this kind of highly complicated equations.}

In the following, we outline conditions under which a tractable valuation is feasible. To do so, we first focus on Asian options by making use of the smoothing effect of the average value such options depend on. Second, we consider more general payoffs and rely on machine learning methods for the evaluation. {These two methodologies represent some first results which allow to use the class of non-linear affine processes with path-dependence in practical applications. However, for further research it is necessary to study the regularity conditions of the associated value functions which are required for these results, see also Remark \ref{rem:Regularity1} and \ref{rem:Regularity2}.}

\section{Asian options}  
\label{subsec:beispiele}

For options which are not path-dependent, the payoff is given by $ \Phi(x_T) = h(x(T))$ {for a} bounded {and} measurable function $h$. 
Then $V_t(x_t)=v(t,x(t))$ for some function $v:[0,T] \times \cX \to \R$ since the distribution of $X(T)$ only depends on $x_t$ through the value $x(t)$. 
If $v$ is continuous (for example if $h$ is Lipschitz continuous as in Lemma 3.6 in \cite{FadinaNeufeldSchmidt2019}), then $v$ is a viscosity solution of the path-independent Kolmogorov align and we obtain Theorem 4.1 in \cite{FadinaNeufeldSchmidt2019} as a special case. 
More interesting in the context we consider here are path-dependent options.

\subsection{Asian options}
\label{subsubsec:asiatische_optionen}

Asian options are European options on the average of a stock price over a certain time interval. Due to the smoothing effect of 
the average, Asian options are  attractive for investors, since the payoff is more difficult to manipulate, and they are easier to hedge and
typically less risky. The history of papers on Asian options is long and we refer for example to \cite{bayraktar2011pricing,kirkby2020efficient} for further references.

{In the following, an} Asian option is considered as a typical path-dependent payoff which illustrates our methodology. To this end, assume that the payoff $\Phi$ is given by 
\begin{align*}
    \Phi(x_T) = h\left(\int_0^T x(s)ds, x(T)\right), \quad x \in C([0,T],\cX)
\end{align*}
for a bounded and measurable function $h$. As above, the value function depends on $x_t$ only through $x(t)$ such that
\begin{align*}
V_t(x_t) = v\left(t,\int_0^t x(s)ds, x(t)\right), \quad (t,x) \in [0,T] \times C([0,T],\cX)
\end{align*}
with some appropriate function $v$.  To compute $v$, we assume that $h$ is sufficiently regular to {guarantee} that $v$ is sufficiently differentiable and consider 
\begin{align}
\label{eq:asiatische_option_pde}
\begin{split}
-\partial_t v(t,y,z) - z\partial_y v(t,y,z) - G(z, \partial_z v(t,y,z), \partial_{zz} v(t,y,z)) &= 0 \quad \quad \quad \text{on } [0,T) \times (\cX \cup \{0\}) \times \cX,\\
v(T,y,z) &= h(y,z) \quad \forall (y,z) \in \cX \times \cX,
\end{split}
\end{align}
where $G$ is given in \eqref{eq:DefinitionNonlinearity}.

\begin{proposition}
\label{lemma:asiatische_option}
Assume $\Theta$ is proper and 
$v \in C^{1,1,2}([0,T] \times (\cX \cup \{0\}) \times \cX)$
with Lipschitz continuous and uniformly bounded derivatives. Then $v$ is a solution of the non-linear PDE in Equation \eqref{eq:asiatische_option_pde}.
\end{proposition}
\begin{proof}
We can extend $V$ to a non-anticipative functional on the space of c\`adl\`ag paths by setting
$V_t(x_t) := v(t,\int_0^t x(s)ds,x(t))$ for $(t,x) \in [0,T] \times D([0,T], \cX).$
Then, by using the notation and definitions introduced in Subsection \ref{subsection:FunctionalDerivatives} we get
\begin{align}
\label{eq:asiatische_option_ableitungenI}
\cD_t V(x_t) &= \lim_{h \searrow 0} \frac1h (V_{t+h}(x_{t,h}) - V_t(x_t))\notag\\
&= \lim_{h \searrow 0} \frac1h \left(v\left(t+h,\int_0^t x(s)ds + x(t)h,x(t)\right) - v\left(t,\int_0^t x(s)ds,x(t)\right)\right)\notag\\
&= \partial_t v\left(t,\int_0^t x(s)ds,x(t)\right) + x(t) \partial_y v\left(t,\int_0^t x(s)ds,x(t)\right).
\end{align}
Moreover, 
\begin{align}
\nabla_x V_t(x_t) 
&=\partial_z v\left(t,\int_0^t x(s)ds,x(t)\right)\quad \text{ and } \quad 
\nabla_x^2 V_t(x_t) 
= \partial_{zz} v\left(t,\int_0^t x(s)ds,x(t)\right). \label{eq:asiatische_option_ableitungenII}
\end{align}

Now we show $V \in \text{Lip}_b^{1,2}([0,T))$. To this end, note that $V \in \C_b^{{1,2}}([0,T))$ by the above computations and the assumptions on $v$.

Thus, we concentrate on the pathwise Lipschitz continuity in \eqref{def:pfadweise_lipschitzstetig} of the derivatives of $V$. Let $x, \tilde{x} \in D([0,T],\cX)$, $t \in [0,T]$ and $h \in [0,T-t)$. We have 
\begin{align*}
| \nabla_x V_{t+h}(x_{t+h}) - \nabla_x V_t(\tilde{x}_t) |
&= \bigg| \partial_z v\Big( t+h,\int_0^{t+h} x(s)ds,x(t+h)\Big) - \partial_z v\Big( t,\int_0^t \tilde{x}(s)ds,\tilde{x}(t)\Big) \bigg|\\
&\leq C \Big( h + \int_0^{t}|x(s)- \tilde{x}(s)|ds + \int_t^{t+h}|x(s)|ds + |x(t+h) - \tilde{x}(t)| \Big)\\
&\leq C \Big( h + h\sup_{s \in [0,h]}|{x}(t+s)| + \sup_{s \in [0,h]}|x(t+s) - \tilde{x}(t)| \Big)\\
&\leq C \Big(  (1+|\tilde{x}(t)|)h + (1+T)\sup_{s \in [0,h]}|x(t+s) - \tilde{x}(t)| \Big)\\
&\leq C \Big( h + \sup_{s \in [0,h]}|x(t+s) - \tilde{x}(t)| \Big),
\end{align*}
with a constant $C$ depending on $x_t,\tilde{x}_t$; again we change the constant $C$ from row to row. The pathwise Lipschitz continuity of $\nabla_x^2 V$ follows in a similar way. Next, we prove that also $\cD_t V$ is pathwise Lipschitz continuous. This follows for the first addend in \eqref{eq:asiatische_option_ableitungenI} by similar arguments as for $\nabla_x V, \nabla_x^2 V$.  For a short notation we set 
\begin{align*}
\Xi_t^{\tilde{x}} &:= \left(t,\int_0^t \tilde{x}(s)ds,\tilde{x}(t)\right) \quad \text{and } 
\quad \Xi_{t+h}^x:= \left(t+h,\int_0^{t+h} x(s)ds,x(t+h)\right).
\end{align*}
For the second addend in \eqref{eq:asiatische_option_ableitungenI} we obtain 
\begin{align*}
\left| x(t+h) \partial_y v(\Xi_{t+h}^x) - \tilde{x}(t)\partial_y v(\Xi_t^{\tilde{x}}) \right| 
&\leq |\partial_y v(\Xi_{t+h}^x)||x(t+h)-\tilde{x}(t)| + |\tilde{x}(t)||\partial_y v(\Xi_{t+h}^x)-\partial_y v(\Xi^{\tilde{x}}_t)|\\
&\leq C \sup_{s \in [0,h]}|x(t+s)-\tilde{x}(t)| + C \, |\tilde{x}(t)| \Big(h + \sup_{s \in [0,h]}|x(t+s) - \tilde{x}(t)| \Big)\\
&\leq C\Big(h + \sup_{s \in [0,h]}|x(t+s) - \tilde{x}(t)| \Big),
\end{align*}
where we use that $\partial_y v$ is bounded and with a changing constant $C$ depending on $x_t,\tilde{x}_t$. Thus, the Lipschitz continuity of $\cD_t V$ follows. By \eqref{eq:asiatische_option_ableitungenII} and the assumption that the derivatives of $v$ are uniformly bounded, it follows that also $\nabla_x V$ is uniformly bounded. Therefore, we conclude that $V \in \text{Lip}_b^{{{1,2}}}([0,T)).$ Hence, $V$ is a  solution of \eqref{eq:kolmogorov}, i.e. 
\begin{align*}
-\partial_t v(\Xi_t^x) - x(t)\partial_y v(\Xi_t^x) - G\Big(x(t),\partial_z v(\Xi_t^x),\partial_{zz}v(\Xi_t^x)\Big) = 0.
\end{align*}
By choosing $x\in C([0,T],\cX)$ appropriately we obtain for all  $(t,y,z) \in (0,T) \times \cX \times \cX$ that 
\begin{align*}
-\partial_t v(t,y,z) - x(t)\partial_y v(t,y,z) - G\Big(x(t),\partial_z v(t,y,z),\partial_{zz}v(t,y,z)\Big) = 0
\end{align*}
which can be extended by continuity to $[0,T)$. It is easy to check the boundary condition and we obtain that $v$ solves 
\eqref{eq:asiatische_option_pde}.
\end{proof}
{
\begin{remark}[Regularity of the value function] \label{rem:Regularity1}
To verify examples under which the assumption regarding the regularity of the value function in Proposition \ref{lemma:asiatische_option} is satisfied, we refer to the enormous literature of Hamilton-Jacobi-Bellman equations. The PDE in Equation \eqref{eq:asiatische_option_pde} with the non-linear function $G$ in \eqref{eq:DefinitionNonlinearity} belongs to this class of equations. Furthermore, we refer to Appendix C in \cite{peng2019nonlinear} which also discusses conditions under which a viscosity solution of a non-linear PDE of this type is a strong solution.
\end{remark}}
The result of Proposition \ref{lemma:asiatische_option} opens the door to the fast numerical evaluation of Asian options under parameter uncertainty. Consider the (bounded) put option given by 
\begin{align*}
\Phi(x_T) = \left( K_1 - \frac1T \int_0^T x(s)ds \right)^+ \wedge K_2, \quad x \in C([0,T]).
\end{align*}
with $K_1, K_2 \in \R$. 
Since $K_2$ can be chosen arbitrarily large it does not play a role in the numerical computation and can be neglected. 

Following the non-linear Vasi\v cek model proposed in \cite[Section 4.2]{FadinaNeufeldSchmidt2019} we choose the model parameters
\begin{align}
\label{eq:asiatische_option_nichtlineares_vasicek_modell_parameter}
\begin{split}
\ubar{\gamma}= \obar{\gamma}=\frac{1}{2},\ \ubar{b}^0 = 0.05,\ \obar{b}^0 = 0.15,\ \ubar{b}^1 &= -3,\ \obar{b}^1 = -0.5,\ \ubar{a}^0 = 0.01,\ \obar{a}^0 = 0.02,\ \ubar{a}^1 = \obar{a}^1 = 0,\\
&T = 1,\ K_1 = 0.2,\ K_2 = 10^6.
\end{split}
\end{align}
As comparison we use the following classical Vasi\v cek models: for model 1 we consider the lower bound of the mean reversion level and for model 2 the upper bound, respectively. Moreover, in both models the volatility equals the upper bound of the non-linear  Vasi\v cek model. This results in the the following model parameters:
\begin{align}
\label{eq:asiatische_option_vasicek_modell_parameter}
\begin{split}
\ubar \gamma=\obar \gamma= \frac{1}{2}, \ \ubar{b}^0 = \obar{b}^0 = 0.15,\ \ubar{b}^1 = \obar{b}^1 = -3,\ \ubar{a}^0 = \obar{a}^0 = 0.02,\ \ubar{a}^1 = \obar{a}^1 = 0 \quad &\text{(Vasi\v cek 1)},\\
\ubar \gamma=\obar \gamma= \frac{1}{2}, \ \ubar{b}^0 = \obar{b}^0 = 0.15,\ \ubar{b}^1 = \obar{b}^1 = -0.5,\ \ubar{a}^0 = \obar{a}^0 = 0.02,\ \ubar{a}^1 = \obar{a}^1 = 0 \quad &\text{(Vasi\v cek 2)}.
\end{split}
\end{align}
The fully non-linear Kolmogorov PDE can be solved with a finite-difference method in Matlab, the outcome is shown in Figure \ref{fig:asiatische_put_option_0_T1} for $x_0 \in [-0.5,0.5]$.
\begin{figure} 
\includegraphics[width=11cm,keepaspectratio]{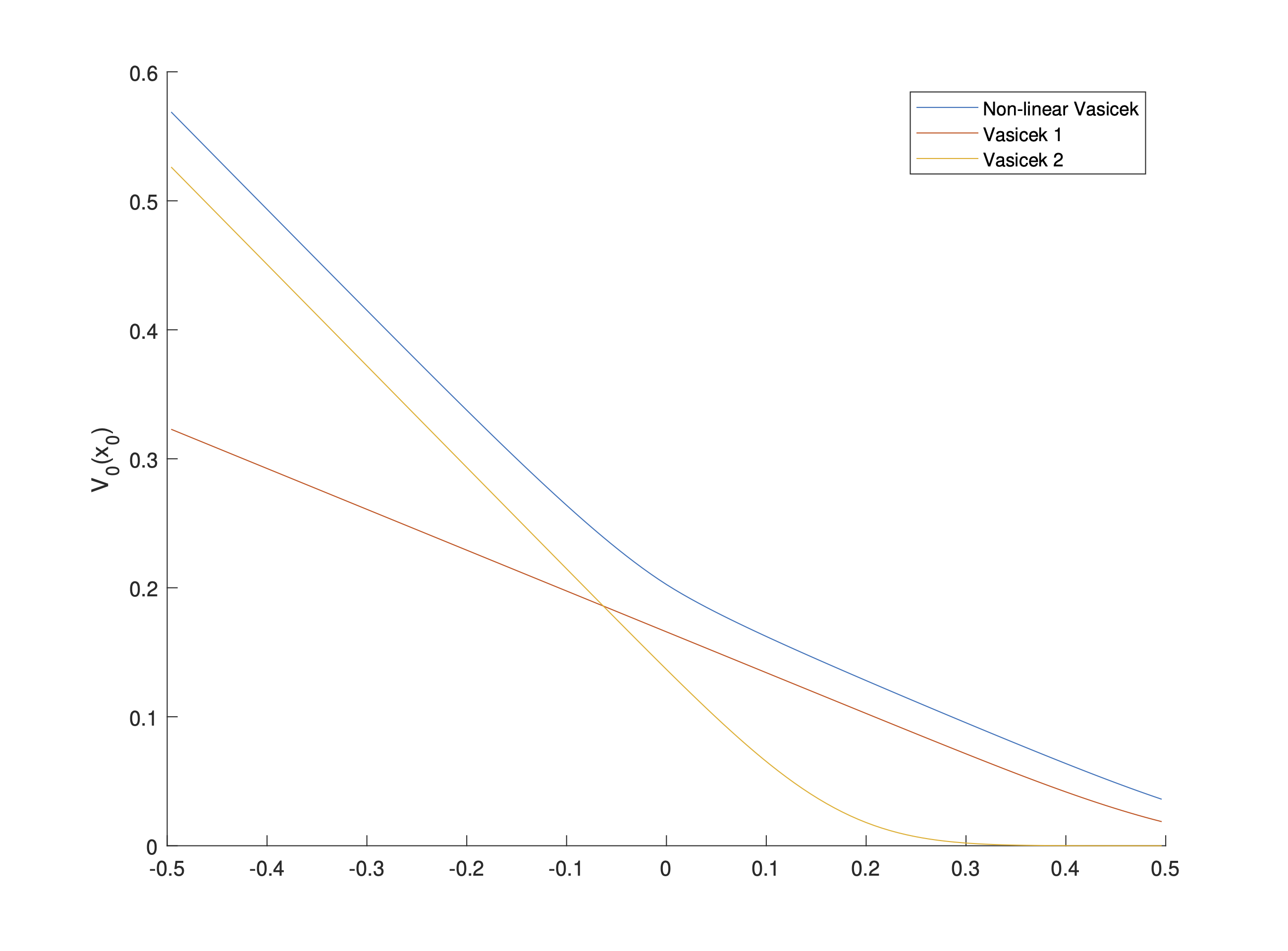}
\caption{Numerically obtained prices for \textbf{Asian put options}. Blue is the non-linear Vasi\v cek model (the upper line) and red and yellow are prices for Vasi\v cek models without uncertainty, however parameters chosen in such a way that the mean reversion speed is either maximal or minimal, compare Equations \eqref{eq:asiatische_option_nichtlineares_vasicek_modell_parameter} and \eqref{eq:asiatische_option_vasicek_modell_parameter}, respectively. On the $x$ axis we plot the initial value $x_0$.}
\label{fig:asiatische_put_option_0_T1} 
\end{figure}  

{Note that} the initial value $x_0$ influences the value of the Asian put in such a way that a higher initial value implies on average a higher value of the integral $\int_0^T X(s)ds$, leading to lower prices for the put. 

Let us first compare the two classical Vasi\v cek models. 
The only difference is in the parameter $b_1$. However, for an intuition it is better to consider the mean reversion speed $\kappa=-b_{1}$ and the mean reversion level $\theta=\nicefrac{b_0}{-b_1}$. Note that by changing $b_1$ both models have different mean-reversion speed \emph{and} mean-reversion level. In the model \textbf{Vasi\v cek 1}, $\kappa  = 3$ and $\theta = 0.05$, while in the model \textbf{Vasi\v cek 2}, $\kappa = 0.5$ and $\theta = 0.3$. In the latter model, the mean-reversion level is at the upper boundary of the interval for $x_0$ (0.5) such that the process has the tendency to increase and the option shows a strong dependence from $x_0$ a little further away from $0.5$ (approximately $0.2$ in our case). In Vasi\v cek 1, the model mean-reverts strongly to $0.05$ and the dependence on the initial value is therefore less pronounced.   

The price under model under uncertainty (the \textbf{non-linear Vasi\v cek model}) clearly dominates the prices without parameter uncertainty. Moreover, the value function shows a smoothed kink around $0$ which stems from the switching behaviour of the supremum in the non-linear Kolmogorov equation at $0$.

\section{The numerical solution of functional PDEs} \label{sec:BarierOptions}

A full numerical study of the non-linear pricing of path-dependent options is far beyond the scope of this article. Our intention is to present a small numerical example which on the one side highlights the feasibility of the chosen numerical approach via deep neural networks and on the other side also shows the challenges of this approach. 

While in the case of Asian options finite difference methods could be applied since the average has a smoothing effect, this will no longer be possible in other cases, for example when Barrier options are considered.  For such options we propose to solve the equation numerically relying on machine learning methods. 
The idea is to write the problem as a forward-backward SDE which we then discretize in time. The resulting functional derivatives are approximated with neural networks.

The proposed method is based on the recent work of \cite{EHanJentzen2017,BeckEJentzen2019}. 

\subsection{Barrier options}
\label{subsubsec:barrier_optionen}
Barrier options are inherently path-dependent. We consider for simplicity digital barrier options, and choose the case of an up-and-in option. These products offer the payoff $1$ if the barrier $B {\geq 0}$ is reached in the interval of consideration, $[0,T]$, and $0$ otherwise. In practice, such products are highly attractive because they often allow for cheaper investments (an up-and-in Call might be significantly cheaper in comparison to a standard Call). More precisely, the payoff is given by 
\begin{align} \label{eq:barrier}
\Phi(x_T) = \Ind_{\{\sup_{s \in [0,T]}x(s) \geq B\}}, \quad x \in C([0,T],\cX){.}
\end{align}
In this case we are not able to reduce the value function to a classical PDE for $v$ since the non-anticipative functional 
\begin{align*}
    F_t(x_t) = \sup_{s \in [0,t]}x(s)
\end{align*}
is not vertically differentiable (consider $t > 0$ and $x$ with  ${\sup_{s \in [0,t]}x(s) = x(0) = x(t)}$).

We therefore propose an algorithm how to solve the (non-linear) functional Kolmogorov PDE in an approximative way by relying on machine learning techniques. Note that a Monte-Carlo simulation in this context is highly intractable because one has to simulate under all measures from $\aset$.

Moreover, solving the Kolmogorov equation \eqref{eq:kolmogorov} directly through numerical methods like finite differences faces the challenge that pathwise derivatives have to be approximated by finite differences which requires a fine discretization of the path space  $D([0,T],{\mathcal{X}})$.
We begin by establishing the connection between functional PDEs and  \emph{forward-backward stochastic differential equations} (FBSDEs).

\subsection{Backward SDEs}
\label{subsec:Ã¼bersicht_Ã¼ber_den_algorithmus}

For an overview on backward stochastic differential equations we refer, for example, to \cite{zhang2017backward}. 
We are interested in functional PDEs of the form
\begin{equation}
\label{eq:allgemeine_funktionale_pde}
\begin{split}
\cD_t F(x_t) &= f\left(t,x(t),F(x_t),\nabla_x F_t(x_t),\nabla_x^2F_t(x_t)\right) \quad \text{on } [0,T) \times C([0,T],{\mathcal{X}}),\\
F_T(x_T) &= g(x_T) \quad \forall x \in C([0,T], {\mathcal{X}})
\end{split}
\end{equation}
with a non-anticipative functional $F$, $T>0$ and mappings 
$f : [0,T] \times \R^4 \to \R$ and $g : C([0,T], {\mathcal{X}}) \to \R.$
Clearly, the non-linear Kolmogorov equation \eqref{eq:kolmogorov} falls into this class.

We consider a filtered probability space $(\Omega,\ccF,\bbF,P)$ together with a Brownian motion $W$. {In the} following, we set $\nabla_x F_T = \nabla_x^2 F_T := 0$.

\begin{lemma}
\label{lemma:numerik_bsde}
Assume that the non-anticipate functional $F \in \mathbb{F}_l^\infty$ satisfies $F \in \mathbb{C}_b^{1,2}([0,T))$ and solves \eqref{eq:allgemeine_funktionale_pde}. {Let $X=(X(t))_{t \in [0,T]}$ be the unique strong solution of the SDE 
$$
dX(t)=b(X(t))dt + \sigma(X(t))dW(t), \quad X(0):=x,
$$
with $b,\sigma: \mathbb{R} \to \mathbb{R}$ and $x \in \mathbb{R}$
and define the process $Y=(Y(t))_{t \in [0,T]}$ by}
\begin{align*}
Y(t) := F_t(X_t), \qquad \quad t \in [0,T].
\end{align*}
Then $Y$ solves the BSDE
{
\begin{align}\label{eq:BSDE}
Y(t) &= g(X_T) - \int_t^T  \nabla_xF_s(X_s)dX(s)  \\
&- \int_t^T \left(f\left(s,X(s),Y(s),\nabla_xF_s(X_s),\nabla_x^2F_s(X_s)\right) + \frac12 \sigma^2 (X(s))\nabla_x^2F_s(X_s) \right )ds, 
\end{align}
$P${-a.s.} for all $t \in [0,T].$}
\end{lemma}
\begin{proof}
Observe that $F$ satisfies the assumptions of the functional It\^o-formula, Theorem 4.1 in  \cite{ContFournie2013}. By this, we obtain that 
{
\begin{align*}
Y(t) & = Y(0) + \int_0^t \left( \cD_s F(X_s) + \frac12 \sigma^2(X(s)) \nabla_x^2 F_s(X_s) + b(X(s))\nabla_xF_s(X_s)\right)ds \\
& \quad+ \int_0^t \sigma(X(s))\nabla_x F_s(X_s) dW(s).
\end{align*}}
Since $F$ solves  \eqref{eq:allgemeine_funktionale_pde} by assumption, we obtain for all  $t \in [0,T)$,
{
\begin{align*}
Y(t) &= Y(0) + \int_0^t \sigma(X(s)) \nabla_xF_s(X_s)dW(s) \\
&+ \int_0^t \left(f\left(s,X(s),Y(s),\nabla_xF_s(X_s),\nabla_x^2F_s(X_s)\right) + \frac12 \sigma^2(X(s)) \nabla_x^2F_s(X_s) + b(X(s)) \nabla_x F_s(X_s) \right)ds, 
\end{align*}
$P$-a.s.}
By continuity, this can be extended to $[0,T]$. 
Hence, for all $t \in [0,T]$, 
\begin{align*}
Y(t) &= Y(T) - (Y(T) - Y(t)) = \eqref{eq:BSDE}. \qedhere
\end{align*}
\end{proof}
{
\begin{remark}[Regularity of the value function] \label{rem:Regularity2}
	In general it is hard to find conditions under which the non-anticipative functional $F$ satisfies the regularity assumptions in Lemma \ref{lemma:numerik_bsde}. However, in the very recent work \cite{cosso_gozzi_rosestolatio_russo_2023} in Appendix B some conditions for the existence of a strong solution for smoothed second-order path-dependent Hamilton-Jacobi-Bellman equations are discussed.
\end{remark}}

\subsection{Deep learning of fractional gradients}

\begin{algorithm}[t] 
\SetAlgoLined
\SetKwInOut{Input}{Input}
\SetKwInOut{Output}{Output}

Consider a discretization in time, $
0 = t_0 < \dots < t_N = T$. The BSDE \eqref{eq:BSDE} is approixmated as follows
\begin{align}\label{eq:numerik_bsde_approximation}
Y(t_{n+1}) &\approx Y(t_n) + \nabla_x F_{t_n}(X_{t_n}){(X(t_{n+1}) - X(t_n))} \notag \\
&+ \Big( f\left(t_n,X(t_n),Y(t_n),\nabla_x F_{t_n}(X_{t_n}),\nabla_x^2 F_{t_n}(X_{t_n})\right) + \frac12 {\sigma^2(X(t_n))} \nabla_x^2 F_{t_n}(X_{t_n}) 
\end{align}
with
$Y(t_N) = g(X_{t_N})$. We proceed as follows:
\begin{enumerate}
\item Choose random starting values $Y(t_0)$, $\nabla_x F_{t_0}(X_{t_0})$ and $\nabla_x^2 F_{t_0}(X_{t_0})$.
\item Approximate $\nabla_x F_{t_n}(X_{t_n})$ and $\nabla_x^2 F_{t_n}(X_{t_n})$, $n=1,\dots,N-1$, via neural networks.
\item Repeat the following training step for the computation of a minibatch of size $M$ until a sufficient level of accuracy is reached:
\begin{itemize}
\item Sample $M$ random paths $(X(t_0),\dots,X(t_N))$.
\item Use \eqref{eq:numerik_bsde_approximation} to compute $m=1,\dots,M$ paths $(Y^m(t_0),\dots,Y^m(t_N))$. 
\item Compute the aggregated loss $\sum_{m=1}^M(Y^m(t_N) - g(X^m_{t_N}))^2$.
\item Update parameter of the neural net to minimize the loss function (for example by using stochastic gradient descent).
\end{itemize}
\end{enumerate}
\caption{Outline of the Algorithm}\label{algo}
\end{algorithm}

As a next step we detail how neural networks can be used for deep-learning of the required gradients $\nabla_x F_{t}(X_{t})$ and $\nabla_x^2 F_{t}(X_{t}).$ Consider a discrete time grid $0 = t_0 < \dots < t_N = T$. Since the solution $Y$ of \eqref{eq:BSDE} satisfies
$$ E_P\big[(Y_T - g(X_T))^2 \big] = 0, $$
a discretized approximation $Y^\theta$ can be {constructed} by minimizing
$$ \theta \mapsto E_P\big[(Y^\theta_N - g(X_{t_0},\dots,X_{t_N}))^2 \big]. $$
 The associated loss function is $\ell(\theta) = (Y^\theta_N - g(X_{t_0},\dots,X_{t_N}))^2$.
The value $g(X_T)$ is approximated by the piecewise constant path of $X$ based on $X_{t_0},\dots,X_{t_N}$.
We shortly denote this in the following by 
$$g(X^N_{T}):=g(X_{t_0},\dots,X_{t_N}), $$
and similarly for other functions depending on the discretized approximation $X^N$ of $X$.  

The parameter vector $\theta\in \R^\nu, \nu \in \mathbb{N},$ contains initial values and the parameters of the neural networks and will be specified step by step.

\subsubsection*{Initial values} 
The first three values  $(\theta_1,\theta_2,\theta_3)$ denote the initial values where $\theta_1=Y_0^\theta$ approximates $Y_{t_0}$, $\theta_2$ approximates $\nabla_x F_{t_0}(X_{t_0})$ and $\theta_3$ approximates $\nabla_x^2 F_{t_0}(X_{t_0}).
$

\subsubsection*{Functional derivatives}
Furthermore, we will use the neural networks $\cNN^{1}_n=\cNN^{1,\theta}_n$ as
approximations of $\nabla_x F_{t_n}(X_{t_n})$, $n=1,\dots,N-1$; the neural networks $\cNN^{2}_n=\cNN^{2,\theta}_n$ approximate $\nabla_x^2 F_{t_n}(X_{t_n})$.

More precisely, $\ccN_n^{1}:\R^n \to \R$ depends on the training parameters $\theta_{j_{n-1}+1},\dots,\theta_{i_n}$ while   $\ccN_n^{2}:\R^n \to \R$ depends on the training parameters $\theta_{i_n+1},\dots,\theta_{j_n}$ with $j_0=3$.

Then, our approximation $Y^N$ is given by (compare Equation \eqref{eq:BSDE})  $Y^N_0 := \theta_1$ and 
\begin{align*}
Y_{n+1}^N &:= Y_n^N + \cNN_n^{1}(X^N_{t_n}) \cdot {(X(t_{n+1}) - X(t_n))} \\
&\quad + \Big( f(t_n,X(t_n),Y_n^N,\cNN_n^{1}(X^N_{t_n}),\cNN_n^{2}(X^N_{t_n})) + \frac12 {\sigma^2(X(t_n))}\cNN_n^{2}(X^N_{t_n}) , 
\end{align*}
$n=1,\dots,N-1$. 

\subsubsection*{Optimization}
As mentioned above the goal of the algorithm is to find $\theta$ such that the expected loss is minimized. To achieve this a \emph{batch} of size $M$ of simulations is constructed and for each simulation $m$, the loss is given by
$$ \ell_m(\theta) := (Y^{m,\theta}_N - g(X^m_{t_0},\dots,X^m_{t_N}))^2. $$
The expectation is estimated by the aggregated loss
$$ \ell(\theta) = \sum_{m=1}^M \ell_m(\theta). $$
Now the parameter $\theta$ can be minimized as detailed in Algorithm \ref{algo} by, for example, stochastic gradient descent. 
The targeted approximation of $Y(0)=F_0(X_0)$ is the parameter $\theta_1$.

\subsection{Numerical results}

Our experiments showed that if the network is chosen too deep, the parameter $\theta_1$ (which is the one which we are mainly interested in) will be less efficiently trained or even not trained at all. For our goals, neural networks with 4 layers and rectified linear units (ReLU) as activation function achieved the best results.

For the numerical results of pricing up-and-in digital options (compare Equation \eqref{eq:barrier}) we consider the case already studied in the previous section, where we chose the parameters according to Equation \eqref{eq:asiatische_option_nichtlineares_vasicek_modell_parameter}.  

\begin{table}
\centering
\begin{tabular}{ r c  c  c  c  c  c  c  c }
\toprule
$x(0)$ &  & -0.3 & -0.2 & -0.1 & 0 & 0.1 & 0.2 & 0.3 \\
\midrule
\multicolumn{2}{l }{mean}     & 0.607 & 0.665 & 0.732 & 0.790 & 0.856 & 1.000 & 1.000 \\
\multicolumn{2}{l }{std.dev.} & 0.008 & 0.006 & 0.005 & 0.004 & 0.002 & 0.003 & 0.003 \\
\bottomrule
\end{tabular}
\ \\[2mm]
\caption{Computed  prices of the \textbf{barrier up-and-in digital options} (10 runs) with associated standard deviations for different initial values $x(0)$ of the non-linear affine process $X$. The model parameters are specified in Equation
\eqref{eq:asiatische_option_nichtlineares_vasicek_modell_parameter}. In most cases the optimal value is already achieved after 25.000 steps, the values here are computed with 100.000 steps, however.}
\label{table:barrier_option}
\end{table}
The outcome of the numerical deep learning algorithm is shown in Table \ref{table:barrier_option}. The algorithm computes, as described above the value $V_0(x_0)$ of the digital up-and-in barrier option as detailed in \eqref{eq:barrier}. Since the outcome of the algorithm is random, we simulated it 10 times and show the mean together with the standard deviation. As expected, prices are increasing in $x_0$. For a high initial value we face small difficulties as for $x(0)=0.2$ the value should  be exactly equal to $1$. For even higher values, the price remains similar.

\begin{figure}
\includegraphics[width=8cm,keepaspectratio]{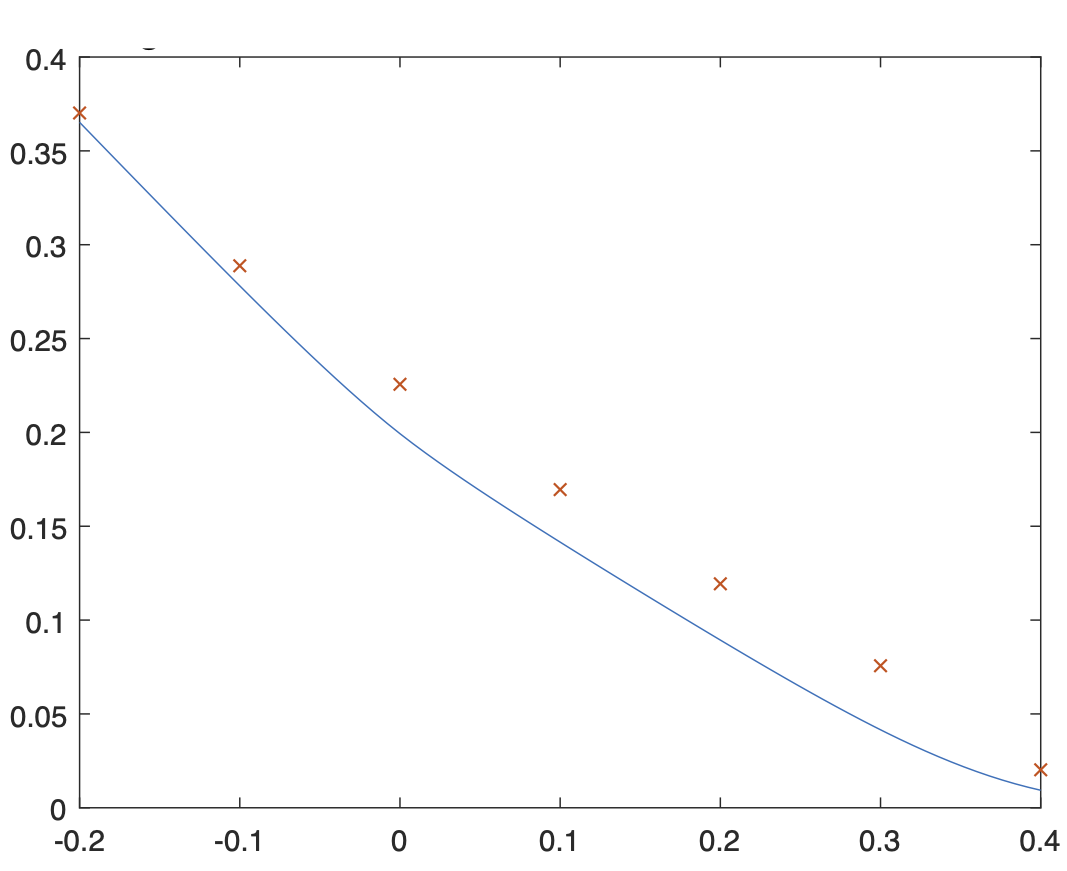}
\caption{Comparison of the finite difference method (solid line, blue) with the deep-learning approach (crosses) in the context of \textbf{Asian put options} according to the model detailed in Equation \eqref{eq:vergleich_modell_parameter}.}
\label{fig:vergleich}
\end{figure}

\subsubsection*{Benchmarking}
To obtain a comparison with existing methods we revisit the Asian put and compare the results from the machine-learning algorithm with the results from finite difference schemes, which are shown in Figure \ref{fig:vergleich}. The chosen parameters are: 
\begin{align}
\label{eq:vergleich_modell_parameter}
\ubar \gamma=\obar \gamma= \frac{1}{2}, \ \ubar{b}^0 = 0.05,\ \obar{b}^0 = 0.15,\ \ubar{b}^1 &= -3,\ \obar{b}^1 = -0.5,\ \ubar{a}^0 = 0.01,\ \obar{a}^0 = 0.02,\ \ubar{a}^1 = \obar{a}^1 = 0,
\end{align} 
with $T = 0.5$ and $K_1 = 0.2,\ K_2 = 10^6$.

We can make the following observations:
\begin{enumerate}
    \item The machine-learning algorithm produces slightly higher values in comparison to the finite difference scheme. The typical kink at $x(0)=0$ is very well visible. The difference is more pronounced in the middle of the observation area, i.e.\ between $0.1$ and $0.3$. 
    \item Clearly, the approximation of $g(X_T)$ with a finite number of gridpoints and piecewise constant extrapolation will cause difficulties when this approximation fails to reach $g(X_T)$. In the cases we consider here, this does not cause problems. 
\end{enumerate}

\section{Further applications} \label{sec:Applications_Credit}
The goal of this section is to highlight the generality of the results in Sections \ref{subsec:beispiele} and \ref{sec:BarierOptions} by spelling out
applications which require path-dependency. We will focus on credit and counterparty risk along with a brief connection to climate risk as examples.

\subsection{Uncertainty in credit risk}
In practical applications of models for credit risk, parameters are  not known and have to be estimated from data. Since defaults are  rare events, data is scarce and thus parameter uncertainty becomes very important. Amongst many other works, parameter uncertainty in credit risk models has been highlighted in \cite{tarashev2010measuring}. 
For a background on credit risk and the models we are using here we refer to \cite{MFE, SchmidtStute2004}.

\smallskip

Conceptually, there are two  different approaches to credit risk: the structural approach and the reduced-form  approach (the latter is also often referred to as intensity-based approach). 
To the best of our knowledge, credit risk under parameter uncertainty has only been studied within the reduced-form framework, see e.g. \cite{bz_2019}, \cite{biagini2020reduced} or \cite{fadina_schmidt_2019}. While in this setting path-dependence is less important, it is fundamental in structural approaches. Therefore,  the  tools presented in Sections \ref{subsec:beispiele} and \ref{sec:BarierOptions} are crucial for entering this field and also allow to handle structural credit risk models in a numerically efficient way, as we will show now. 

\subsubsection{The Merton model} \label{subsec:Merton_model}
We start with the {case of no path-dependence} and consider an extension of the famous Merton model, introduced in \cite{Merton1974}. In this framework, default occurs only at maturity $T$ and hence can not be placed into an intensity-framed approach as shown in \cite{GehmlichSchmidt2016MF}.

The main driver in this model is {the} firm value of the considered company. In the Merton framework itself {this} is a geometric Brownian motion with known parameters. Taking uncertainty into account and generalizing this setting substantially, let us assume 
that the firm value follows a generalized affine process under uncertainty. Recall that a generalized affine process is motivated from diffusion processes following Equation \eqref{eq:affiner_prozess_sde} and thus contains the Merton model as special case (by choosing $b_0=a_0=0$ and $\gamma = 1$). 

In line with this we consider as firm value the  path $x \in \Omega$. We fix a current time $t {\in [0,T]} $ and let our model be represented by all semimartingale laws in $\cA(t,x_t,\Theta)$.
The company's debt is {given} by one outstanding bond with face value $D>0$ and maturity $T>0.$ 
If the capital at time $T$ is insufficient to pay back the outstanding bound, i.e.\ when ${x}(T) < D$, a default occurs. The remaining capital is paid to the bond holders which receive the payoff $(D-x(T))^+$. On the other side, the shareholders receive the remaining value of the firm at time $T$, which  is given by $(x(T)-D)^+$.

Taking the perspective of the bond owner, the value of the associated payoff in the ambiguity framework we consider here might be represented by the infimum over all expectations taken from the considered models in $\cA(t,x_t,\Theta)$. More precisely, the bond holder is interested in the associated value function 
$$V_t(x_t):=\inf_{P \in \mathcal{A}(t,x_t, \Theta)}E_P[(D-{x}(T))^+] = - \sup_{P \in \mathcal{A}(t,x_t, \Theta)}E_P[-(D-{x}(T))^+]  $$ 
which can be represented by $V_t(x_t)=v(t,x(t))$ for some function $v: [0,T] \times \mathcal{X} \to \mathbb{R}$. Under continuity conditions for $v$ we can apply the results in \cite{FadinaNeufeldSchmidt2019}, see Section \ref{subsec:beispiele} for more details. In particular, this allows to represent the value function with a non-linear PDE, which can be used for numerical simulations as in \cite{FadinaNeufeldSchmidt2019}.

\subsubsection{The Black-Cox model} \label{subsec:BlackCox_model}
Since in the Merton model default does only occur at $T$, \cite{black_cox_1976} proposed a first-passage time model where  the default {takes place} at the \emph{first time} when the firm value drops under a pre-specified threshold. 
This approach is inherently path-dependent  and therefore the simpler methods in the Merton model, as proposed in Section 
\ref{subsec:Merton_model}, can no longer be applied which we will show in the following.

As above the firm value is given by the path $x\in \Omega$ and {we consider} model under uncertainty from starting time $t{\in (}0{,T]}$ on {represented by} the generalized affine model under uncertainty, ${\cA}(t,x_t,\Theta)$. 
Let $(D(t))_{t {\in[0,T]}}$ be {a sequence of deterministic and time-dependent thresholds.} The default time $\tau$ in the Black-Cox model is  the predictable stopping time given by 
$$
\tau:=\inf \lbrace t \geq 0: x(t) \leq  D(t) \rbrace {\wedge T}.
$$

Considering a {put} option in this setting requires taking default into account. While in reality default might not lead to a complete loss of the capital, for simplicity we consider zero recovery for stock investors. 
{In this case} the payoff of a {put} option on the stock price with strike $K>0$ and maturity $T$ is {given} by
\begin{equation} \label{eq:PayoffBlackCoxModel}
g(x_T):={(K-x(T))^+} \ind{ \tau > T }
= {(K-x(T))^+}\ind{ x(t) > D(t) \colon t \in [0,T]} .
\end{equation}
This is a highly path-dependent payoff function and for $D(t){\equiv} D$ with $D>0$ it can be seen as a combination of a call option with the payoff given in \eqref{eq:barrier}. The methods in Section \ref{sec:BarierOptions} allow to numerically approximate the associated value function 
$$
V_t(x_T):=\sup_{P \in \mathcal{A}(t,x_t, \Theta)} E_P[g(x_T)]
$$
of the payoff in \eqref{eq:PayoffBlackCoxModel}. Here, uncertainty is {taken} into account by computing the worst-case price via the supremum.

\subsection{Climate-risk}
The Black-Cox approach also provides a method to assess climate transition risk\footnote{This idea was propagated in \cite{janosik2023} which will be published in a forthcoming paper.}. In such an approach one challenges the firm value by substracting a climate-related damage which we assume in the first step to be deterministic and denote it by the function $t \mapsto C(t)$.  

The resulting climate-stressed default time is then given by 
$$
\tau^*:=\inf \lbrace t \geq 0: x(t) - C(t) \leq  D(t) \rbrace {\wedge T}=
\inf \lbrace t \geq 0: x(t)  \leq  D(t)+C(t) \rbrace {\wedge T},
$$	
which shows that this setting can be embedded in the previous framework by replacing the default barrier $D$ with $C+D$. Moreover, if there is uncertainty on the climate impact $C$, due to monotonicity, one could directly solve this problem by considering the worst-case climate impact.

\subsection{Bilateral counterparty risk and a multi-dimensional approach} \label{subsec:Counterparty_risk}
We now consider a framework which allows to take into account bilateral counterparty risk, as for example in Section 6.2 in \cite{huyen_2010}. To this end, it is necessary to consider a multi-dimensional extension of our framework as already pointed out in Remark \ref{rem:multidim}.

To do so, we study multi-dimensional non-linear generalized affine processes following \cite{biagini_bollweg_oberpriller_2022}. Indeed, note that  by setting $\Omega:=C([0,T], \mathbb{R}^n)$ for $n >1$ we  arrive at a $n$-dimensional non-linear generalized affine process with history $x=(x^1,....,x^n)$ until time $t$ by modeling the appropriately adjusted semimartingale characteristics of the $n$-dimensional canonical process $X$ on $\Omega$, as follows. Fix a non-empty and closed state space $\textnormal{M}\subseteq\mathbb{R}^n$.
    Let $b=(b_0,\ldots,b)\in(\mathbb{R}^n)^{n+1}$, $a=(a_0,\ldots,a_n)\in\mathbb{S}^{n+1}$, {where, $\mathbb{S}$ denotes the vector space of symmetric $(n \times n)$-matrices.}
    Define for $y=(y^1,\ldots,y^n)^T \in \mathbb{R}^n$ the following functions	
    \begin{align}
    	b(y)&:=\Big(b_0+(b_1,\ldots,b_n)\,y\Big)\,\mathbf{1}_{\textnormal{M}}(x)
    	=\Big(b_0+\sum_{i=1}^n y^i b^i \Big)\,\mathbf{1}_{\textnormal{M}}(y) \in \mathbb{R}^n, \label{eq:theta1}\\
    	a(y)&:=\Big(a_0+(a_1,\ldots,a_n)\,y\Big)\,\mathbf{1}_{\textnormal{M}}(y) 
    	=\Big(a_0 + \sum_{i=1}^n {y^i} a_i\Big)\,\mathbf{1}_{\textnormal{M}}(y) \in \mathbb{S}.\label{eq:theta2}
     \end{align}
    We refer to $\theta=(b,a)
    \in(\mathbb{R}^n)^{n+1}\times\mathbb{S}^{n+1}$ as \emph{parameter}, and define for each parameter the map
    \begin{equation}\label{eq:theta}
\mathbb{R}^n\rightarrow\mathbb{R}^n\times\mathbb{S},\quad   x\mapsto\theta(x):=(b(x),a(x)),
    \end{equation}
    with $a$, $b$ in \eqref{eq:theta1} - \eqref{eq:theta2}.
Similarly for a subset $\Theta\subseteq
    (\mathbb{R}^n)^{n+1}\times\mathbb{S}^{n+1}$ and $x\in\mathbb{R}^n$, define
    \begin{equation}\label{eq:Theta}
    	\Theta(x):=\Big\{\theta(x)\,:\,\theta\in\Theta\Big\}\subseteq\mathbb{R}^n\times\mathbb{S}.
    \end{equation}
    A set $\Theta$ of parameters is \emph{closed}, if it is closed with respect to the topology which makes $(\mathbb{R}^n)^{n+1}\times\mathbb{S}^{n+1}$ a separable metric space. For $x \in \Omega, t \geq 0$, we define the set of probability measure{s} $\mathcal{A}(t,x_t,\Theta)$ as follows:
    \begin{equation}\label{eq:NLAJD}
        \Big\{P\in\mathfrak{P}_{\operatorname{sem}}^{\operatorname{ac}}(t)\,:\,\;P(X_t=x_t)=1;\:(\beta^P(s), \alpha(s)^P)\in\Theta(X_s)\text{ for }dt\otimes dP\text{-a.e.} (\omega, s) \in \Omega \times (t,T]\Big\}.
    \end{equation}
    For $\Theta\subseteq(\mathbb{R}^n)^{n+1}\times\mathbb{S}^{n+1}$ non-empty and closed, the tuple $(X,\{\mathcal{A}(t,x_t,\Theta)\}_{x\in\textnormal{M}},\textnormal{M})$, where $X$ is the canonical process and $\mathcal{A}(t,x_t,\Theta)$ is defined in \eqref{eq:NLAJD}, is called \emph{non-linear generalized affine process with history $x=(x^1,....,x^n)$ until time $t$} with parameter set $\Theta$ and state space $\textnormal{M}$.

For simplicity we consider the two-dimensional case with  $d=2$ and study a portfolio consisting of two firms whose corresponding asset values (first without any bilateral counterparty risk) $x=(x^1, x^2)$ are modeled as a two-dimensional generalized affine process under uncertainty. 

The main goal is to capture contagion effects, i.e.\ the response of company 2 on a default of company 1 and vice versa. To achieve this we first consider the default time without bilateral counterparty risk and then introduce the effect of the default on the surviving company. More specifically, in this setting firm $i$ defaults at time $\tau_i$ and the value of firm $i$ decreases if company $j$ defaults due to a contagion effect. For more than two companies this procedure would have to be iterated. 

Let $(D^i(t))_{t \geq 0}$, $i=1,2$ be deterministic (time-dependent) default thresholds and define the default time of company $i$ (without taking contagion into account)  by
$$
\tau^i:=\inf \lbrace t \geq 0: x_t^i \leq D^i(t) \rbrace {\wedge T}.
$$

By $e^{i,j}>0$ for $i=1,2, \ i \neq j$ we denote the proportional effect on the value of firm $i$ in the case where firm $j$ defaults. Let $\bar{x}:=(\bar{x}^1, \bar{x}^2)$ be the firm's value taking into account the external counterparty default as well as the contagion given by
\begin{align}
    \bar{x}^i_t&:=x^i_t \ind{\tau_i > t, \tau_j > t} + x_t^i  (1-e^{i,j}) \ind{\tau_2 \leq t < \tau_1} \nonumber \\
    &= x^i_t  \ind{x_t^i > D^i(t), x_t^j > D^j(t) \text{ for all }t \in [0,T] } + x_t^i (1-e^{i,j})  \ind{x_t^i > D^i(t) \text{ for all }t \in [0,T]} \ind{t \geq \tilde{t} \text{ with } x^j(\tilde{t})= D^j (\tilde{t})}. \nonumber
\end{align}

The value function of a {put} option on firm $i$ which takes into account the risk of a default and the contagion effect is then given by
\begin{equation}\label{eq:bilateral}
V_t(x_T):= \sup_{P \in \mathcal{A}(t,x_t, \Theta)} E_P[{(K-\bar{x}^i(T))^+}].
\end{equation}
In this case, the approach with neural networks as introduced in Section \ref{sec:BarierOptions} opens the door to derive the price even for a payoff function which is strongly path-dependent.

In general, there is a large uncertainty on the height of the counterparty effects $e^{i,j}$ but, due to the present monotonicity, considering an uncertainty interval $[\underbar e^{i,j}, \bar e^{i,j}]$ for these effects and assuming that they are independent of the firm values, one could directly use the upper bound for computing the worst-case expectation in \eqref{eq:bilateral}.

\begin{appendix}

\section{Semimartingale comparison} \label{app:PO}
For the convenience of the reader we shortly recall the main result in \cite{Bergenthum_Rueschendorf_2007} and the results from Appendix 1 in \cite{FadinaNeufeldSchmidt2019}. 
For a real-valued Markov process $S^*$ and a terminal time $T$ we define the propagation operator
$$ \cG^g(t,x) := E[g(S^*_{T})|S^*_t=x]. $$
 Stability under the propagation operator (PO) will be a crucial property.

\begin{definition}
	For some function class $\cF$ and some Markov process $S^*$, we say that PO$(S^*,\cF)$ holds if $\cG^g(t,\cdot) \in \cF$ for all $0 \le t \le T$ and for all $g \in \cF$.
\end{definition}

Consider a second Markov process $S$, possibly on a different probability space, and denote by
\begin{align*} 
	\cF_{icx}:=\{ f:\R \to \R, \text{increasing and convex} \}.
\end{align*}

Proposition 11 in \cite{FadinaNeufeldSchmidt2019} now implies the propagation of increasing and convex functions for suitable Markov processes:
\begin{proposition}\label{prop:PO}
Let $S^*$ be a strong and homogeneous Markov process with continuous sample paths sucht that    
\begin{align} \label{eq:42}
          E[S^*_{t}|S^*_0=x]=\pi_0(t) +  x \pi_1(t), \qquad 0 \le t \le T
\end{align} holds with $\pi_1(t)\neq 0$ for all $t \in [0,T]$.
Then, PO$(S^*,\cF_{icx})$ holds.
\end{proposition}

It was also observed  that \eqref{eq:42} holds for affine processes. More precisely, when $E[e^{uX_t}|X_0=x] = e^{\phi(t) + \psi(t)x}$,   an application of Theorem 13.2 in \cite{JacodProtter} yields that
\begin{align} 
\label{eq:moments}
	E[X_t|X_0=x] = \partial_u \phi(t)+ x \partial_u \psi(t); 
\end{align}
the expressions for $\phi$ and $\psi$ are available in Section 10.3.2 in  \cite{Filipovic2009}.

\end{appendix}

\bibliographystyle{agsm}
\bibliography{Non_linear_Affine_processes}

@string{jof = {Journal of Finance}}

@string{mf = {Mathematical Finance}}

@string{springer = {Springer Verlag. Berlin Heidelberg New York}}

@book{GikhmanSkorokhoddIII,
  title={The theory of stochastic processes III},
  author={Gikhman, Iosif Ilʹich and Skorokhod, Anatoli{\u\i} Vladimirovich},
  year={1979},
  publisher={Springer, New York}
}

@article{tarashev2010measuring,
  title={Measuring portfolio credit risk correctly: Why parameter uncertainty matters},
  author={Tarashev, Nikola},
  journal={Journal of Banking \& Finance},
  volume={34},
  number={9},
  pages={2065--2076},
  year={2010},
  publisher={Elsevier}
}

@article{janosik2023,
  title={Measuring Climate Impact and Climate Risk and the Connection to
the Probability of Default},
  author={Reka Janosik and Patrick Jahn},
  journal={MathFinance Conference, Frankfurt, Presentation},
  year={2023},
}

@book{zhang2017backward,
  title={Backward {S}tochastic {D}ifferential {E}quations},
  author={Zhang, Jianfeng},
  booktitle={Backward Stochastic Differential Equations},
  pages={79--99},
  year={2017},
  publisher={Springer}
}

@article{EHanJentzen2017,
  title={Deep learning-based numerical methods for high-dimensional parabolic partial differential equations and backward stochastic differential equations},
  author={W. E and J. Han and A. Jentzen},
  journal={Communications in Mathematics and Statistics},
  volume={5},
  number={4},
  pages={349--380},
  year={2017},
  publisher={Springer}
}

@article{BeckEJentzen2019,
  title={Machine learning approximation algorithms for high-dimensional fully nonlinear partial differential equations and second-order backward stochastic differential equations},
  author={C. Beck and W. E and A. Jentzen},
  journal={Journal of Nonlinear Science},
  volume={29},
  number={4},
  pages={1563--1619},
  year={2019},
  publisher={Springer}
}

@article{kirkby2020efficient,
  title={Efficient {Asian} option pricing under regime switching jump diffusions and stochastic volatility models},
  author={Kirkby, J Lars and Nguyen, Duy},
  journal={Annals of Finance},
  volume={16},
  number={3},
  pages={307--351},
  year={2020},
  publisher={Springer}
}

@article{bayraktar2011pricing,
  title={Pricing {A}sian options for jump diffusion},
  author={Bayraktar, Erhan and Xing, Hao},   
  journal={Mathematical Finance},
  volume={21},
  number={1},
  pages={117--143},
  year={2011},
  publisher={Wiley Online Library}
}

@article{ContFournie2013,
author="R. Cont and D.-A. Fourni\'e",
title="Functional {I}t{\^o} Calculus and Stochastic Integral Representation of Martingales",
journal="The Annals of Probability",
volume="41",
number="1",
pages="109-133",
year="2013"
}

@article{ElKarouiTan2015,
author="N. {El Karoui} and X. Tan",
title="Capacities, Measurable Selection and Dynamic Programming Part {II}: Application in Stochastic Control Problems",
journal="arXiv:1310.3364v2",
year="2015"
}

@article{FadinaNeufeldSchmidt2019,
  title={Affine processes under parameter uncertainty},
  author={Fadina, Tolulope and Neufeld, Ariel and Schmidt, Thorsten},
  journal={Probability, Uncertainty and Quantitative Risk},
  volume={4},
  number={1},
  pages={1},
  year={2019},
  publisher={Springer}
}

@article{criens2022non,
  title={Non-Linear Continuous Semimartingales},
  author={Criens, David and Niemann, Lars},
  journal={Electronic Journal of Probability},
  volume={28},
  number={1},
  pages={1-40},
  year={2023}
}

@book{peng2019nonlinear,
  title={Nonlinear expectations and stochastic calculus under uncertainty: with robust {C}{L}{T} and {G}-{B}rownian motion},
  author={Peng, Shige},
  volume={95},
  year={2019},
  publisher={Springer Nature}
}

@article{lutkebohmert2021robust,
  title={Robust deep hedging},
  author={L{\"u}tkebohmert, Eva and Schmidt, Thorsten and Sester, Julian},
  journal={Quantitative Finance},
  year={2022},
  volume = {22},
number = {8},
pages = {1465-1480}
}

@article{biagini2020reduced,
  title={Reduced-form setting under model uncertainty with non-linear affine intensities},
  author={Biagini, Francesca and Oberpriller, Katharina},
  journal={Probability, Uncertainty and Quantitative Risk},
  volume={6},
  number={3},
  pages={159-188},
  year={2021}
}

@article{KellerResselSchmidtWardenga2018,
  title={Affine processes beyond stochastic continuity},
  author={Keller-Ressel, Martin and Schmidt, Thorsten and Wardenga, Robert},
  journal={The Annals of Applied Probability},
  volume={29},
  number={6},
  pages={3387--3437},
  year={2019},
  publisher={Institute of Mathematical Statistics}
}

@article{denk2017semigroup,
  title={A semigroup approach to nonlinear {L\'evy} processes},
  author={Denk, Robert and Kupper, Michael and Nendel, Max},
  journal={Stochastic Processes and their Applications},
  volume={130},
  number={3},
  pages={1616-1642},
  year={2020}
}

@article{Grbac2015affine,
  title={Affine LIBOR models with multiple curves: Theory, examples and calibration},
  author={Grbac, Zorana and Papapantoleon, Antonis and Schoenmakers, John and Skovmand, David},
  journal={SIAM Journal on Financial Mathematics},
  volume={6},
  number={1},
  pages={984--1025},
  year={2015},
  publisher={SIAM}
}

@book{JacodProtter,
author = {J. Jacod and P. Protter},
title = {Probability Essentials},
publisher = Springer,
year = {2004},
}

@Article{NeufeldNutz2014,
  Title                    = {Measurability of Semimartingale Characteristics with Respect to the Probability Law},
  Author                   = {Ariel Neufeld and Marcel Nutz},
journal = "Stochastic Processes and their Applications",
  Year                     = {2014},
volume = "124",
number = "11",
pages = "3819--3845",
}

@article{NeufeldNutz2017,
  title={Nonlinear {L}\'{e}vy processes and their characteristics},
  author={Neufeld, Ariel and Nutz, Marcel},
  journal={Transactions of the American Mathematical Society},
  volume={369},
  number={1},
  pages={69--95},
  year={2017}
}

@article{Acciaioetal2016,
  title={A model-free version of the fundamental theorem of asset pricing and the super-replication theorem},
  author={Acciaio, Beatrice and Beiglb{\"o}ck, Mathias and Penkner, Friedrich and Schachermayer, Walter},
  journal={Mathematical Finance},
  volume={26},
  number={2},
  pages={233--251},
  year={2016},
}

@article{bielecki2018arbitrage,
  title={Arbitrage-free pricing of derivatives in nonlinear market models},
  author={Bielecki, Tomasz R and Cialenco, Igor and Rutkowski, Marek},
  journal={Probability, Uncertainty and Quantitative Risk},
  volume={3},
  number={1},
  pages={2},
  year={2018}
}

@book{guyon2013nonlinear,
  title={Nonlinear option pricing},
  author={Guyon, Julien and Henry-Labord{\`e}re, Pierre},
  year={2013},
  publisher={CRC Press}
}

@article{muhle2018risk,
  title={A risk-neutral equilibrium leading to uncertain volatility pricing},
  author={Muhle-Karbe, Johannes and Nutz, Marcel},
  journal={Finance and Stochastics},
  volume={22},
  number={2},
  pages={281--295},
  year={2018},
  publisher={Springer}
}

@Article{Nutz2013,
  Title                    = {Random {$G$}-expectations},
  Author                   = {Nutz, Marcel},
  Journal                  = {Annals of Applied Probability},
  Year                     = {2013},
  Number                   = {5},
  Pages                    = {1755--1777},
  Volume                   = {23},

  Doi                      = {10.1214/12-AAP885},
  Fjournal                 = {The Annals of Applied Probability},
  Publisher                = {The Institute of Mathematical Statistics},
}

@Article{Martini,
  Title                    = {A theoretical framework for the pricing of contingent claims
 in the presence of model uncertainty},
  Author                   = {Denis, L. and Martini, C.},
  Journal                  = {The Annals of Applied Probability},
  Year                     = {2006},
  Number                   = {2},
  Pages                    = {827--852},
  Volume                   = {16},

  Fjournal                 = {The Annals of Applied Probability},
  ISSN                     = {1050-5164},

}

@book {JeanblancChesneyYor2009,
    AUTHOR = {Jeanblanc, Monique and Yor, Marc and Chesney, Marc},
     TITLE = {Mathematical methods for financial markets},
    SERIES = {Springer Finance},
 PUBLISHER = {Springer-Verlag London Ltd.},
   ADDRESS = {London},
      YEAR = {2009},
}

@Article{EberleinGrbacSchmidt2013,
author =       {E. Eberlein and Z. Grbac and T. Schmidt},
title =        {Discrete tenor models for credit risky portfolios driven by time-inhomogeneous {L\'evy} processes},
journal =         {SIAM Journal on Financial Mathematics},
year =         {2013},
volume = {4},
issue = {1},
pages={616--649},
}

@incollection{FilipovicSchmidt2010,
	Author = {D. Filipovi\'{c} and T. Schmidt},
	Booktitle = {Contemporary Quantitative Finance},
	Editor = {C. Chiarella and A. Novikov},
	Pages = {231--254},
	Publisher = Springer,
	Title = {Pricing and Hedging of {CDO}s: A Top-Down Approach},
	Year = 2010}

@article{schrager2006affine,
  title={Affine stochastic mortality},
  author={Schrager, David F},
  journal={Insurance: Mathematics and Economics},
  volume={38},
  number={1},
  pages={81--97},
  year={2006},
  publisher={Elsevier}
}

@article{zeddouk2020mean,
  title={Mean reversion in stochastic mortality: why and how?},
  author={Zeddouk, Fadoua and Devolder, Pierre},
  journal={European Actuarial Journal},
  volume={10},
  number={2},
  pages={499--525},
  year={2020},
  publisher={Springer}
}

@article{biffis2005affine,
  title={Affine processes for dynamic mortality and actuarial valuations},
  author={Biffis, Enrico},
  journal={Insurance: Mathematics and Economics},
  volume={37},
  number={3},
  pages={443--468},
  year={2005},
  publisher={Elsevier}
}

@article{russo2011calibrating,
  title={Calibrating affine stochastic mortality models using term assurance premiums},
  author={Russo, Vincenzo and Giacometti, Rosella and Ortobelli, Sergio and Rachev, Svetlozar and Fabozzi, Frank J},
  journal={Insurance: Mathematics and Economics},
  volume={49},
  number={1},
  pages={53--60},
  year={2011},
  publisher={Elsevier}
}

@article{ErraisGieseckeGoldberg2010,
	Author = {Errais, Eymen and Giesecke, Kay and Goldberg, Lisa R.},
	Doi = {10.1137/090771272},
	Issn = {1945-497X},
	Journal = {SIAM Journal on Financial Mathematics},
	Pages = {642--665},
	Title = {Affine point processes and portfolio credit risk},
	Volume = {1},
	Year = {2010}}

@article {GehmlichSchmidt2016MF,
author = {Gehmlich, Frank and Schmidt, Thorsten},
title = {Dynamic Defaultable Term Structure Modelling beyond the Intensity Paradigm},
journal = {Mathematical Finance},
volume = {28},
number = {1},
pages = {211--239},
year = {2018},
}

@incollection{cuchieroFilipovicTeichmann:ATS,
	Author = {C. Cuchiero and D. Filipovi\'{c} and J. Teichmann},
	Booktitle = {Encyclopedia of Quantitative Finance},
	Editor = {Rama Cont},
	Title = {Affine Models},
	Year = {2009}}

@article{FilipovicSchmidtOverbeck2011, 
	Author = {D. Filipovi\'{c} and L. Overbeck and T. Schmidt},
	Journal = {Mathematical Finance},
	Pages = {53--71},
	Title = {Dynamic {CDO} Term Structure Modelling},
	Volume = 21,
	Year = {2011}}

@article{gumbel2020machine,
  title={Machine learning for multiple yield curve markets: fast calibration in the Gaussian affine framework},
  author={G{\"u}mbel, Sandrine and Schmidt, Thorsten},
  journal={Risks},
  volume={8},
  number={2},
  pages={50},
  year={2020},
  publisher={MDPI}
}

@article{DuffieFilipovicSchachermayer,
	Author = {D. Duffie and D. Filipovi\'{c} and W. Schachermayer},
	Journal = {Annals of Applied Probability},
	Pages = {984--1053},
	Title = {Affine processes and applications in finance},
	Volume = {13},
	Year = {2003}}

@book{Filipovic2009,
	Author = {Damir Filipovi\'{c}},
	Publisher = Springer,
	Title = {Term Structure Models: A Graduate Course},
	Year = {2009}}

@book{MFE,
	Author = {A. McNeil and R. Frey and P. Embrechts},
	Publisher = {Princeton University Press},
	Title = {Quantitative Risk Management: Concepts, Techniques and Tools},
	Year = {2015}}

@article{Merton1974,
	Author = {R. Merton},
	Journal = JoF,
	Pages = {449-470},
	Title = {On the pricing of corporate debt: the risk structure of interest rates},
	Volume = {29},
	Year = {1974}}

@article{TSchmidt_InfiniteFactors,
	Author = {T. Schmidt},
	Journal = {International Journal of Theoretical and Applied Finance},
	Pages = {43-68},
	Title = {An Infinite Factor Model for Credit Risk},
	Volume = 9,
	Year = {2006}}

@article{SchmidtStute2004,
	Author = {T. Schmidt and W. Stute},
	Journal = {Contemporary Mathematics},
	Pages = {75 - 115},
	Title = {Credit Risk -- A Survey},
	Volume = {336},
	Year = {2004}}

@book{BertsekasShreve1978,
author="D. P. Bertsekas and S. E. Shreve",
title="Stochastic Optimal Control: The Discrete Time Case",
publisher="Acadamic Press",
year="1978"
}

@article{cont-06,
	Author = {Cont, R.},
	Journal = mf,
	Pages = {519-542},
	Title = {Model uncertainty and its impact on the pricing of derivative instruments},
	Volume = {16},
	Year = {2006}}

@article{cheridito-et-al-05,
	Author = {Patrick Cheridito and H.Mete Soner and Nizar Touzi and Nicolas Victoir},
	Journal = {Communications in Pure and Applied Mathematics},
	Title = {Second order backward stochastic differential equations and fully nonlinear parabolic {P}{D}{E}s},
	volume={60},
	number={7},
	pages={1081-1110},
	Year = {2007}}

@article{Janson_Tysk_2006,
  title={Feynman-{K}ac Formulas for {B}lack-{S}choles-Type Operators},
  author={Janson, Svante and Tysk, Johan},
  journal={Bulletin of the {L}ondon {M}athematical {S}ociety},
  volume={38},
  number={2},
  pages={269--282},
  year={2006}
}

@article{criens2022MarkovSelection,
  title={Markov selections and {F}eller properties of nonlinear diffusions},
  author={Criens, David and Niemann, Lars},
  journal={Stochastic Processes and their Applications},
  volume={173},
  year={2024}
}

@article{Bergenthum_Rueschendorf_2007,
  title={Comparison of semimartingales and {L\'{e}vy} processes},
  author={Bergenthum, Jan and R\"{u}schendorf, Ludger},
  journal={The {A}nnals of {P}robability},
  volume={35},
  number={1},
  pages={228-254},
  year={2007}
}

@article{huyen_2010,
title={Stochastic control under progressive enlargement of filtrations and applications to multiple defaults risk management},
author={Pham, H.},
journal={Stochastic Processes and their Applications},
volume={120},
number={9},
pages={1795--1820},
year={2010}
}

@article{black_cox_1976,
title={Valuing corporate securities: Some effects of bond indenture provisions},
author={Black, F. and Cox, J.C},
journal={Journal of Finance},
volume={31},
pages={351-367},
year={1976}
}

@article{biagini_bollweg_oberpriller_2022,
title={Non-linear affine processes with jumps},
author={Biagini, Francesca and Bollweg, Georg and Oberpriller, Katharina},
journal={Probability, Uncertainty and Quantitative Risk},
volume={8},
number={2},
pages={235-266},
year={2023}
}

@article{fadina_schmidt_2019, 
title={Default Ambiguity},
journal={Risks},
volume={7},
number={2},
author={Fadina, Tolulope and Schmidt, Thorsten},
year={2019},
DOI={10.3390/risks7020064}}

@article{bz_2019, 
title={Reduced-form framework under model uncertainty}, 
journal={The Annals of Applied Probability}, 
volume={29},
number={4},
author={Biagini, Francesca and Zhang, Yinglin}, 
year={2019}, 
pages={2481-2522},
DOI={10.1214/18-aap1458}}

@article{denis_hu_peng, 
title={Function Spaces and Capacity Related to a Sublinear Expectation: Application to {G}-{B}rownian Motion Paths}, 
journal={Potential Analysis}, 
volume={34},
author={Denis, Laurent and Hu, Mingshang and Peng, Shige}, 
year={2011}, 
pages={139-161}}

@incollection{peng,
  title={{B}{S}{D}{E} and related g-expectations},
  author={Peng, Shige},
  booktitle={Backward {S}tochstic {D}ifferential {E}quations},
   editor = {N. El Karoui and L. Mazliak},
  year={1997},
  pages={141-159},
}

@article{matoussi_possamai_zhou, 
title={Robust Utility maximization in nondominated models with 2{B}{S}{D}{E}: {T}he uncertain volatility model}, 
journal={Mathematical Finance}, 
volume={25},
number={2},
author={Matoussi, Anis and Possama\"{i}, Dylan and Zhou, Chaou}, 
year={2015}, 
pages={258-287}}

@article{soner_touzi_zhang_2012, 
title={Wellposedness of second order backward {S}{D}{E}s}, 
journal={Probability Theory and Related Fields}, 
volume={153},
author={Soner, Mete H. and Touzi, Nizar and Zhang, Jianfeng}, 
year={2012}, 
pages={149-190}}

@article{soner_touzi_zhang_2013, 
title={Dual formulation of second order target problems}, 
journal={The Annals of Applied Probability}, 
volume={23},
number={1},
author={Soner, Mete H. and Touzi, Nizar and Zhang, Jianfeng}, 
year={2013}, 
pages={308-347}}

@article{bartl_kupper_neufeld, 
title={Duality theory for robust utility maximisation}, 
journal={Finance and Stochastics}, 
volume={25},
author={Bartl, Daniel and Kupper, Michael and Neufeld, Ariel}, 
year={2021}, 
pages={469-503}}

@article{dolinsky_soner, 
title={Martingale optimal transport and robust hedging in continuous time}, 
journal={Probability theory and Related fields}, 
volume={160},
author={Dolinsky, Yan and Soner, Mete H.}, 
year={2014}, 
pages={391-427}}

@article{HERNANDEZHERNANDEZ2007980,
title = {A control approach to robust utility maximization with logarithmic utility and time-consistent penalties},
journal = {Stochastic Processes and their Applications},
volume = {117},
number = {8},
pages = {980-1000},
year = {2007},
author = {Daniel Hern\'{a}ndez-Hern\'{a}ndez and Alexander Schied}}

@article{nutz_2015, 
title={Robust superhedging with jumps and diffusion}, 
journal={Stochastic Processes and their Applications}, 
volume={125},
number={12},
author={Nutz, Marcel}, 
year={2015}, 
pages={4543-4555},
DOI={10.1016/j.spa.2015.07.008}}

@article{park_wong_2022, 
title={Robust Consumption-Investment with Return Ambiguity: A Dual Approach with Volatility Ambiguity}, 
journal={SIAM Journal on Financial Mathematics}, 
volume={13},
number={3},
author={Park, Kyunghyun and Wong, Hoi Ying}, 
year={2022}, 
pages = {802-843},
}

@article{schied_2007, 
title={Optimal investments for risk- and ambiguity-averse preferences: a duality approach}, 
journal={Finance and Stochastics}, 
volume={11},
author={Schied, Alexander}, 
year={2007}, 
pages = {107-129},
}

@article{zhou_2023,
author = {Zhou, Jianjun},
title = {{Viscosity solutions to second order path-dependent {H}amilton-{J}acobi-{B}ellman equations and applications}},
volume = {33},
journal = {The Annals of Applied Probability},
number = {6B},
pages = {5564--5612},
year = {2023}
}

@article{cosso_gozzi_rosestolatio_russo_2023,
author={Cosso, Andrea and Gozzi, Fausto and Rosetolato, Mauro and Russo, Francesco},
title={Path-dependent {H}amilton-{J}acobi-{B}ellman equation: Uniqueness of {C}randall-{L}ions viscosity solutions},
journal={arXiv:2107.05959},
year={2023}
}

\end{document}